\documentclass[reqno]{amsart}

\usepackage[utf8x]{inputenc}
\usepackage{amsmath}
\usepackage{amsfonts}
\usepackage{amssymb}
\usepackage{mathrsfs}
\usepackage[]{hyperref}
\usepackage{pgfplots}
\pgfplotsset{compat = newest}

\newtheorem{theorem}{Theorem}
\theoremstyle{plain}

\newtheorem{definition}{Definition}

\newtheorem{lemma}{Lemma}
\newtheorem{proposition}{Proposition}

\numberwithin{equation}{section}

\newcommand{\naturals}{\mathbb{N}}
\newcommand{\reals}{\mathbb{R}}
\newcommand{\complexes}{\mathbb{C}}

\renewcommand{\Re}[1]{\operatorname{Re}{#1}}
\renewcommand{\Im}[1]{\operatorname{Im}{#1}}
\newcommand{\spec}[1]{\operatorname{spec}{#1}}
\newcommand{\proj}[1]{\mathsf{P}_{#1}}
\newcommand{\tr}[1]{\operatorname{tr}{#1}}

\newcommand{\id}[1]{\mathrm{id}_{#1}}
\newcommand{\comm}[2]{[#1,#2]}
\newcommand{\acomm}[2]{\{#1,#2\}}
\newcommand{\adj}{\prime}
\newcommand{\hadj}{*}
\newcommand{\iprod}[2]{\langle #1, #2 \rangle}
\newcommand{\hsiprod}[2]{\langle #1, #2 \rangle_{2}}
\renewcommand{\vec}[1]{\mathbf{#1}}

\newcommand{\hsnorm}[1]{\left\| #1 \right\|_{2}}

\newcommand{\matr}[1]{\mathbb{M}_{#1}(\complexes)}
\newcommand{\matrd}{\matr{d}}
\newcommand{\pos}[2]{\operatorname{\mathbf{P}}(#1,#2)}
\newcommand{\pose}[1]{\operatorname{\mathbf{P}}(#1)}
\newcommand{\cp}[2]{\operatorname{\mathbf{CP}}(#1,#2)}
\newcommand{\cpe}[1]{\operatorname{\mathbf{CP}}(#1)}

\newcommand{\cocpe}[1]{\operatorname{\mathbf{coCP}}(#1)}

\newcommand{\dece}[1]{\operatorname{\mathbf{D}}(#1)}
\newcommand{\transpose}{\mathrm{T}}

\begin{document}
\title[D-divisible quantum evolution families]{D-divisible quantum evolution families}
\author{Krzysztof Szczygielski}
\address[K. Szczygielski]
{Institute of Physics, Faculty of Physics, Astronomy and Informatics, Nicolaus Copernicus University, Grudzi\c{a}dzka 5/7, 87–100, Toruń, Poland}
\email[K. Szczygielski]{krzysztof.szczygielski@umk.pl}%

\begin{abstract}
We propose and explore a notion of \emph{decomposably divisible} (D-divisible) differentiable quantum evolution families on matrix algebras. This is achieved by replacing the complete positivity requirement, imposed on the propagator, by more general condition of decomposability. It is shown that such D-divisible dynamical maps satisfy a generalized version of Master Equation and are totally characterized by their time-local generators. Necessary and sufficient conditions for D-divisibility are found. Additionally, decomposable trace preserving semigroups are examined.
\end{abstract}

\maketitle

\section{Introduction}

The aim of this article is to define, construct and characterize a generalization of CP-divisible (i.e.~Markovian) evolution families, or quantum dynamical maps, on matrix algebras onto a certain subclass of much broader, however still mathematically manageable case of \emph{decomposable positive maps}. We restrict our attention to the case of \emph{decomposably divisible} families, i.e.~such maps $\Lambda_t$ on matrix algebra $\matrd$, which are divisible and which \emph{propagators} are trace preserving and decomposable on $\matrd$. Decomposability is a relatively simple, yet non-trivial generalization of \emph{complete positivity}, which in turn has been a well-characterized and motivated concept in quantum theory since 1970's (see \cite{Alicki2006a,Breuer2002,Rivas2012} and references within), traditionally used to model time evolution of quantum systems. In particular, CP-divisible families \cite{Lindblad1976,Gorini1976} has been granted a special attention, since CP-divisibility is commonly considered equivalent to Markovianity. We abandon this approach here in favor of D-divisibility, effectively obtaining a new subclass of non-Markovian evolution families (or \emph{weakly non-Markovian}, using terminology of \cite{Chruscinski2014}; see also \cite{Chruscinski2022}). We hope that such decomposable dynamical maps might be useful in future for description of physical systems outside a Markovian regime, for example influenced by more sophisticated quantum effects or to mirror the existence of higher-order correlations in the system.

The article is structured as follows. In section \ref{sec:Preliminaries} we provide some mathematical preliminaries, including notion of decomposable maps over algebra of complex matrices, as well as some basic description of dynamics of open quantum systems. The main part of the article is the section \ref{sec:DdivisibleQuantumDynamicalMaps}, devoted to \emph{D-divisible quantum evolution families}, where we formulate a necessary and sufficient conditions for D-divisibility expressed in terms of associated time-dependent generators. Construction of such is presented in Theorem \ref{thm:MainTheorem}, which is our main result. In section \ref{sec:Semigroups} we remark on a semigroup case and present some results related to their asymptotic behavior (Theorems \ref{thm:AsymptCPSemigroupCondition} and \ref{thm:AsymptCPSemigroupConditionSpectrum}). Finally, section \ref{sec:Examples} presents two simple examples in dimension 2 and 3.

\section{Preliminaries}
\label{sec:Preliminaries}

First, we provide some basic preliminaries including notions of decomposability of positive maps and divisibility (and Markovianity) of quantum dynamics. We will be working a lot with \emph{Hilbert-Schmidt} bases spanning space $\matrd$, i.e.~bases orthonormal with respect to the \emph{Hilbert-Schmidt inner product} (also called \emph{Frobenius inner product}) on $\matrd$, given via
\begin{equation}
	\hsiprod{a}{b} = \tr{a^\hadj b} = \sum_{i,j=1}^{d} \overline{a_{ij}}b_{ij}, \qquad a,b \in \matrd .
\end{equation}
Amongst all such bases, one consisting of strictly \emph{Hermitian matrices} will be granted a special attention. Namely, let $\{F_i\}_{i=1}^{d^2}$ be a Hilbert-Schmidt basis subject to constraints
\begin{equation}\label{eq:HSonb}
	F_{i} = F_{i}^{\hadj}, \quad \tr{F_i} = \delta_{i d^2}, \quad F_{d^2} = \frac{1}{\sqrt{d}} I.
\end{equation}
Such basis may be seen as a generalization of both Pauli and Gell-Mann matrices and may be constructed in similar way (see appendix \ref{subsect:HSbasis} for details and for some more properties). By construction, matrices $F_i$ can be either \emph{non-diagonal and symmetric}, \emph{antisymmetric} or \emph{diagonal} (where all $F_i$ s.t.~$i<d^2$ are traceless). For any $d$, there is exactly $d(d-1)/2$ of both symmetric and antisymmetric matrices and $d$ diagonal ones. We reserve symbol $F_i$ for such a basis exclusively throughout the whole article and introduce an accompanying enumeration, such that $F_i$ will be:
\begin{itemize}
	\item \emph{symmetric} for $1 \leqslant i \leqslant \frac{1}{2}d(d-1)$,
	\item \emph{antisymmetric} for $1 + \frac{1}{2}d(d-1) \leqslant i \leqslant d(d-1)$,
	\item \emph{diagonal} for $1 + d(d-1) \leqslant i \leqslant d^2$.
\end{itemize}
The following composition rule will be of importance: for every $F_i$, $F_j$ we have
\begin{equation}\label{eq:FiFjCompRule}
	F_i F_j = \sum_{k=1}^{d^2} \xi_{ijk} F_k,
\end{equation}
where coefficients $\xi_{ijk}$ may be computed as
\begin{equation}\label{eq:XiCoeff}
	\xi_{ijk} = \hsiprod{F_k}{F_i F_j} = \tr{F_i F_j F_k}
\end{equation}
and are expressible in terms of so-called \emph{structure constants}, which characterize $\matrd$ as a Lie algebra. It is then a simple exercise to check that the following identities hold:
\begin{equation}\label{eq:XiProp}
	\xi_{ijk} = \xi_{kij} = \xi_{jki}, \quad \overline{\xi_{ijk}} = \xi_{jik}.
\end{equation}

\subsection{Decomposable maps}

Let $\mathscr{A}$, $\mathscr{B}$ be ordered, unital *-algebras and let $\mathscr{A}^{+}$, $\mathscr{B}^{+}$ stand for convex cones of positive elements of $\mathscr{A}$ and $\mathscr{B}$ respectively. We say that a bounded linear map $\phi : \mathscr{A} \to \mathscr{B}$ is \emph{positive}, or $\phi\in\pos{\mathscr{A}}{\mathscr{B}}$, if $\phi(\mathscr{A}^{+}) \subseteq \mathscr{B}^{+}$, i.e.~it maps positive elements into positive. Moreover, if an extended map $\phi_n = \id{}\otimes \phi$, acting on $\mathbb{M}_n (\mathscr{A}) \simeq \matr{n} \otimes \mathscr{A}$ via prescription $\phi_n ([a_{ij}]) = [\phi(a_{ij})]$, $a_{ij}\in\mathscr{A}$, is also positive for some $n$, we say $\phi$ is \emph{$n$-positive}; if in addition it is $n$-positive for all $n\in\naturals$, map $\phi$ is called \emph{completely positive} (CP), or $\phi\in\cp{\mathscr{A}}{\mathscr{B}}$. Both sets $\pos{\mathscr{A}}{\mathscr{B}}$, $\cp{\mathscr{A}}{\mathscr{B}}$ are then convex cones in space of all linear maps from $\mathscr{A}$ to $\mathscr{B}$.

Structure of CP maps is characterized by means of the famous \emph{Stinespring dilation theorem} stating that for every $\phi \in \cp{\mathscr{A}}{B(H)}$ for $\mathscr{A}$ a unital C*-algebra and $H$ a Hilbert space, exists some auxiliary Hilbert space $K$ such that $\phi$ admits a (nonunique) representation as a composition
\begin{equation}
	\phi(a) = V^* \pi (a) V, \qquad a\in\mathscr{A},
\end{equation}
for some bounded operator $V : H\to K$ and *-homomorphism $\pi : \mathscr{A}\to B(K)$. If both $\mathscr{A}$ and $H$ in question are finite-dimensional, i.e.~$\phi$ acts between algebras of matrices, $\phi : \matr{n} \to \matr{m}$, one defines the so-called \emph{Choi matrix} of $\phi$, 
\begin{equation}\label{eq:ChoiMatrix}
	C_\phi = \sum_{i,j = 1}^{n^2} E_{ij} \otimes \phi(E_{ij}),
\end{equation}
where $E_{ij}$ are matrix units (i.e.~they contain 1 in place $(i,j)$ and 0s everywhere else) spanning $\matr{n}$. Mapping $\phi \mapsto C_\phi$ is a bijection from $B(\matr{n}, \matr{m})$ into $\matr{n}\otimes\matr{m}\simeq \matr{mn}$ known as the \emph{Choi-Jamio\l{}kowski isomorphism}. Then, Stinespring dilation theorem is equivalent to the famous \emph{Choi's theorem} \cite{Choi_1975}, which stays that $\phi$ is CP iff (if and only if) it is $n$-positive, which is then true iff $C_\phi \in \matr{mn}^+$. Furthermore, as a corollary, it can be shown that for every $\phi \in \cp{\matr{n}}{\matr{m}}$ there exists a set of matrices $\{X_i\}_{i=1}^{mn} \subset \matr{m\times n}$ such that
\begin{equation}
	\phi(a) = \sum_{i=1}^{mn} X_{i} a X_{i}^{*}, \quad a\in\matr{n},
\end{equation}
which is the \emph{Kraus decomposition} of $\phi$ (matrices $X_i$ are called \emph{Kraus operators}) associated with $\phi$. The notion of complete positivity proved itself to be very robust concept, both in mathematics and physics. Unfortunately, although the complete characterization of CP maps is known due to results by Stinespring, Choi and Kraus, we lack such in case of merely positive maps and finding it has been a long-standing goal in mathematics for many years.
\vskip\baselineskip
Throughout this paper, we will be focusing on a special sub-class of positive maps, the so-called \emph{decomposable maps}, which may be seen as a conceptually simple, however still nontrivial generalization of CP maps. Moreover, from now on we assume all maps under consideration to be exclusively \emph{endomorphisms} over matrix slgebra $\matrd$ and we tweak our notation accordingly by writing simply $B(\matrd)$, $\pose{\matrd}$ and $\cpe{\matrd}$ for appropriate maps on this algebra.
\vskip\baselineskip
Let $\theta : \matrd \to \matrd$ denote the \emph{transposition map}, i.e.
\begin{equation}
	\theta (a) = a^{\transpose}, \quad [a_{ij}] \mapsto [a_{ji}],
\end{equation}
with respect to some chosen basis in $\complexes^d$. It is easy to see that $\theta$ is a \emph{positive map}, however it is not CP (in fact, it fails to be even 2-positive). Transposition allows to define yet another class of positive maps, the so-called \emph{completely copositive maps}. One says that a map $\phi \in \pose{\matrd}$ is \emph{completely copositive} (coCP), if its composition with $\theta$ is CP, or that there exists some $\tilde{\phi}\in\cpe{\matrd}$ such that
\begin{equation}
    \phi = \theta \circ \tilde{\phi}.
\end{equation}
The marriage of notions of both complete positivity and copositivity determines a class of \emph{decomposable maps}, which will remain at our focus throughout this article:

\begin{definition}
Let $\varphi \in \pose{\matrd}$. We say $\varphi$ is \emph{decomposable}, $\varphi \in \dece{\matrd}$, if it can be expressed as a convex combination of CP and coCP map, i.e.~if there exist $\phi, \psi \in \cpe{\matrd}$ such that
\begin{equation}\label{eq:DecomposableMapDecomposition}
	\varphi = \phi + \theta \circ \psi .
\end{equation}
\end{definition}
Decomposable maps may be also characterized in terms of a following necessary and sufficient condition. Let $\varphi\in\pose{\matrd}$ and let $C_\varphi \in \matr{d^2}$ be its corresponding Choi matrix. By identification $\matr{d^2} \simeq \matrd \otimes \matrd$ we introduce a linear map of \emph{partial transposition} (with respect to second factor) $\Gamma$ on $\matr{d^2}$, defined by its action on simple tensors as
\begin{equation}
	a\otimes b \mapsto (a\otimes b)^{\Gamma} = a \otimes b^{\transpose}.
\end{equation}
Define also two convex cones
\begin{equation}
	V_d = \matr{d^2}^{+}, \quad V^{\Gamma}_{d} = \{\rho : \rho^{\Gamma} \in \matr{d^2}^{+}\}.
\end{equation}
Then, a following characterization of decomposable maps applies \cite{Woronowicz1976,Chruscinski2018}:

\begin{theorem}
\label{thm:DecCondition}
Map $\varphi$ on $\matrd$ is decomposable iff
\begin{equation}
	\forall \, \rho\in V_{d}\cap V_{d}^{\Gamma} : \tr{C_\varphi \rho} \geqslant 0.
\end{equation}
\end{theorem}

In practice, verifying if a given linear map is decomposable by finding exact decomposition into a combination \eqref{eq:DecomposableMapDecomposition} of its CP and coCP part may be a hopeless task, even in low dimensional algebras. Instead, condition stated in theorem \ref{thm:DecCondition} can be checked quite sufficiently by means of a semidefinite programming (SDP) routines, as is also the case in this article.

Every decomposable map $\varphi$ is in addition \emph{Hermiticity preserving}, i.e.~it satisfies
\begin{equation}
	\varphi (a)^\hadj = \varphi (a^\hadj)
\end{equation}
for all $a\in\matrd$. It is known from works by St{\o}rmer and Woronowicz \cite{Woronowicz1976,Stormer_1963} that cones of positive and decomposable maps in $B(\matr{n},\matr{m})$ are equal if $mn\leqslant 6$, i.e.~every positive map is decomposable in such case; in particular, all positive endomorphisms on $\matr{2}$ are decomposable, as well as positive maps between $\matr{2}$ and $\matr{3}$. The question of exact conditions for decomposability in higher-dimensional algebras remains unanswered, however counter-examples are known in literature already for maps on $\matr{3}$.

\subsection{Quantum evolution families}

Here we provide some basic description of evolution in theory of open quantum systems. Let $\rho_t$ stand for a time-dependent \emph{density matrix} of some $d$-dimensional quantum system, i.e.~let
\begin{equation}
	\rho_t \in \matrd^{+}, \quad \tr{\rho_t} = 1 \qquad \text{for all } \, t\in\reals_+ .
\end{equation}
A family of linear, time-parametrized maps $\{\Lambda_t : t\in\reals_+ \}$ on $\matrd$, providing an evolution of density matrix via equation
\begin{equation}
	\rho_t = \Lambda_t (\rho_0)
\end{equation}
for some initial $\rho_0$, will be called the \emph{quantum evolution family}, or \emph{quantum dynamical map}. In order to maintain the probabilistic interpretation of $\rho_t$ as density matrix at every $t\geqslant 0$, it is required for $\Lambda_t$ to be \emph{trace preserving} (i.e.~$\tr{\Lambda_t (\rho)} = \tr{\rho}$) and \emph{positive}. By physical reasoning, one often demands not merely a positivity, but rather complete positivity of $\Lambda_t$ (one can find appropriate explanation e.g. in \cite{Alicki2006a,Breuer2002,Rivas2012} and numerous other sources). This restriction, however, will be abandoned in this paper in favor of \emph{decomposability}.
\begin{definition}
We say that quantum evolution family $\{\Lambda_t : t\in\reals_+ \}$ is \emph{divisible} in some interval $[t_1, t_2] \subseteq \reals_+$ if for every $t \in [t_1, t_2]$ and every $s \in [t_1, t]$ there exists a map $V_{t,s}$ satisfying
\begin{equation}\label{eq:qdpropagator}
	\Lambda_t = V_{t,s} \circ \Lambda_s .
\end{equation}
If in addition $V_{t,s}$ is a \emph{positive} or \emph{completely positive} map for every $s \leqslant t$, then $\{\Lambda_t : t\in\reals_+ \}$ is called \emph{P-divisible} or \emph{CP-divisible} in this interval, respectively. 
\end{definition}
Such two-parameter family of maps $\{V_{t,s} : s\leqslant t\}$ is then called the \emph{propagator} of evolution family (as $V_{t,s}$ \emph{propagates} $\Lambda_s$ forward in time). If $\Lambda_t$ is invertible then it is immediate that $V_{t,s} = \Lambda_t \circ \Lambda_{s}^{-1}$. CP-divisibility is commonly identified with Markovianity.

It is most frequently assumed, that the dynamical map in question satisfies the time-local \emph{Master Equation} in two equivalent forms
\begin{equation}\label{eq:ME}
	\frac{d\Lambda_t}{dt} = L_t \circ \Lambda_t \qquad \text{or} \qquad \frac{d\rho_t}{dt} = L_t (\rho_t),
\end{equation}
for some map $L_t \in B(\matrd)$, called a \emph{generator}. All dynamical maps obeying \eqref{eq:ME} are divisible. By celebrated results of Lindblad, Gorini, Kossakowski and Sudarshan \cite{Lindblad1976,Gorini1976}, a necessary and sufficient condition for an invertible map $\Lambda_t$ subject to Master Equation \eqref{eq:ME} to be \emph{CP-divisible} is that $L_t$ must be of a form
\begin{equation}\label{eq:standardForm}
	L_t (\rho) = -i \comm{H_t}{\rho} + \sum_{j,k=1}^{d^2 -1} a_{jk}(t) \left( F_j \rho F_{k} - \frac{1}{2}\acomm{F_{k}F_j}{\rho}\right),
\end{equation}
where $H_t$ is Hermitian and $[a_{jk}(t)]\in\matr{d^2 -1}^+$ for all $t\in\reals_+$ ($\acomm{a}{b} = ab+ba$ is the anticommutator). Equation \eqref{eq:standardForm} defines so-called \emph{standard form} (also \emph{Lindblad form} or \emph{LGKS form}) of $L_t$. On physics grounds, $H_t$ is identified with system's Hamiltonian (which includes Lamb-shift corrections; here one puts $\hbar = 1$ for brevity) and matrix $[a_{jk}(t)]$, being commonly called the \emph{Kossakowski matrix}, expresses the ``non-unitary'' part of the evolution due to interactions between system and the environment. If generator $L_t$ is time-independent, i.e.~$L_t = L$, then a solution of Master Equation \eqref{eq:ME} is a one-parameter contraction semigroup $\{e^{tL} : t\in\reals_+\}$ of trace preserving CP maps, known as the \emph{Quantum Dynamical Semigroup}.

\section{D-divisible quantum evolution families}
\label{sec:DdivisibleQuantumDynamicalMaps}

\subsection{Notion of D-divisibility} 
\label{sec:NotionOfDdivisibility}

In this section we propose and elaborate on the notion of \emph{D-divisibility}. Let $\{\Lambda_t : t\in\reals_+\}$ again stand for a family of positive and trace preserving maps on $\matrd$. Then, we define D-divisibility of this family in a manner analogous to CP-divisibility by demanding that the propagator is decomposable:

\begin{definition}
We say that a family $\{\Lambda_t : t\in\reals_+\}$ is D-divisible (decomposably divisible) in interval $[t_1, t_2] \subseteq\reals_+$, iff it is divisible in $[t_1, t_2]$ and its associated propagator $V_{t,s}$ is trace preserving and decomposable for all $s,t \in [t_1, t_2]$, $s\leqslant t$, i.e.
\begin{equation}
	V_{t,s} = X_{t,s} + \theta \circ Y_{t,s},
\end{equation}
for some maps $X_{t,s}, Y_{t,s} \in \cpe{\matrd}$ continuously depending on $(t,s)$.
\end{definition}
We stress here that although map $V_{t,s}$ is required to be trace preserving as a whole, neither of maps $X_{t,s}$, $Y_{t,s}$ is \emph{a priori} expected to be so:
\begin{proposition}\label{prop:Vproperties}
Let a family $\{\Lambda_t : t \in\reals_+\}$ be D-divisible in $[t_1, t_2] \subseteq\reals_+$ and let $\Lambda_0 = \id{}$. Then, the following hold for all $t\in[t_1, t_2]$ and all $s \in [t_1, t]$:
\begin{enumerate}
	\item \label{item:VpropId} $V_{t,t} = \id{}$,
	\item \label{item:VpropXttId} $X_{t,t} = \id{}$,
	\item \label{item:VpropYttZero} $Y_{t,t} = 0$,
	\item \label{item:VpropDecTP} $\Lambda_t\in\dece{\matrd}$ and is trace preserving,
	\item \label{item:VpropTP} $X_{t,s} + Y_{t,s}$ is trace preserving.
\end{enumerate}
\end{proposition}

\begin{proof}
Property \ref{item:VpropId} follows immediately from divisibility condition \eqref{eq:qdpropagator} after taking $s=t$. As a consequence $V_{t,t}$ is a  decomposable map with its coCP part being zero, so properties \ref{item:VpropXttId} and \ref{item:VpropYttZero} follow. For property \ref{item:VpropDecTP}, see that \eqref{eq:qdpropagator} also yields $\Lambda_t = V_{t,0} \circ \Lambda_0 = V_{t,0}$ and so $\Lambda_t$ is decomposable and trace preserving. Remaining property \ref{item:VpropTP} then follows from linearity of trace and trace preservation of transposition map after simple algebra. 
\end{proof}

\subsection{Generators of decomposable dynamics}
\label{sec:GeneratorsOfDecomposableDynamics}

In this section we present our main result, i.e.~a necessary and sufficient condition for a quantum evolution family to be D-divisible expressed in terms of properties of the associated generator. Before that we briefly discuss some additional notions. Our construction of generator (given in a proof of theorem \ref{thm:MainTheorem}) will be heavily depending on so-called \emph{operator sum representation} of linear maps on $\matrd$, including the transposition map. Namely, if $T$ is any linear endomorphism on algebra $\matrd$, its action on $a\in\matrd$ may be always represented in a form
\begin{equation}
	T(a) = \sum_{i,j=1}^{d^2} t_{ij} F_i a F_j
\end{equation}
for some matrix of coefficients $[t_{ij}]\in\matr{d^2}$. In addition, $T\in\cpe{\matrd}$ iff $[t_{ij}]\in\matr{d^2}^+$. Similarly, the transposition map $\theta$ admits an operator-sum representation of a form
\begin{equation}\label{eq:ThetaRep}
	\theta (a) = a^{\transpose} = \sum_{i=1}^{d^2} \theta_i F_i a F_{i}
\end{equation}
for coefficients $\theta_i \in \{-1, \, 1\}$ given as
\begin{equation}\label{eq:ThetaIdef}
	\theta_i = \begin{cases} -1, \quad \text{for } 1+\frac{1}{2}d(d-1) \leqslant i \leqslant d(d-1), \\ +1, \quad \text{otherwise.} \end{cases}
\end{equation}
Proof of this statement is available in the appendix \ref{subsect:OperatorSumRep}. We will use coefficients $\theta_i$ given above to define a particular 4-index \emph{geometric tensor}, which will be of crucial importance later on. Recall that basis matrices $F_i$ obey composition rule \eqref{eq:FiFjCompRule} for coefficients $\xi_{ijk} = \tr{F_i F_j F_k}$.
\begin{definition}
We define the 4-index \emph{geometric tensor} $\hat{\mathbf{\Omega}} = [\Omega_{\mu\nu}^{jk}]$, where $1 \leqslant j,k \leqslant d^2$, $1\leqslant\mu\nu\leqslant d^2-1$, by setting
\begin{equation}\label{eq:OmegaTensor}
	\Omega_{\mu\nu}^{jk} = \sum_{i=1}^{d^2} \theta_i \xi_{ij\mu} \overline{\xi_{ik\nu}}.
\end{equation}
\end{definition}
One can easily show (see lemma \ref{lemma:OmegaAlternateForm} in section \ref{app:AdditionalResults} of the Appendix) that $\mathbf{\hat{\Omega}}$ admits a somewhat more compact and robust representation as
\begin{equation}\label{eq:OmegaTensorRep}
	\Omega_{\mu\nu}^{jk} = \hsiprod{F_{k}^{\transpose} F_\mu}{F_{\nu}^{\transpose}F_j},
\end{equation}
which will become useful. Now we are ready to formulate our main result:

\begin{theorem}\label{thm:MainTheorem}
Let a family $\{\Lambda_t : t \in \reals_+ \}$ of maps on $\matrd$ satisfy an ordinary differential equation
\begin{equation}\label{eq:LambdaODE}
	\frac{d\Lambda_t}{dt} = L_t \circ \Lambda_t, \quad \Lambda_0 = \id{},
\end{equation}
where $L_t \in B(\matrd)$ and function $t\mapsto L_t$ is continuous everywhere in interval $[t_1, t_2]\subseteq\reals_+$. Then, family $\{\Lambda_t : t \in \reals_+\}$ is D-divisible and trace preserving in this interval iff there exists a map $M_t$ on $\matrd$ in standard form, Hermitian matrix $K_t \in \matrd$ and matrix $[\omega_{jk}(t)] \in \matr{d^2}^+$ such that
\begin{equation}\label{eq:LtDecomposition}
    L_t = M_t + N_t, \quad t \in [t_1, t_2],
\end{equation}
where $N_t$ admits a form
\begin{equation}\label{eq:NtDecomposition}
	N_t (\rho) = - i \comm{K_{t}}{\rho} + \sum_{\mu,\nu =1}^{d^2 -1}\eta_{\mu\nu}(t) \left( F_\mu \rho F_{\nu} - \frac{1}{2}\acomm{F_{\nu}F_{\mu}}{\rho} \right)
\end{equation}
for coefficients
\begin{equation}\label{eq:EtaMatrixOmega}
	\eta_{\mu\nu} (t) = \sum_{j,k=1}^{d^2} \Omega_{\mu\nu}^{jk} \omega_{jk}(t).
\end{equation}
\end{theorem}

\begin{proof}
The proof will follow general guidelines of \cite[Theorem 4.2.1]{Rivas2012}. We are interested in computing $\frac{d\rho_t}{dt}$, where the derivative is to be calculated ``from above'', i.e.
\begin{equation}
	\frac{d\rho_t}{dt} = \lim_{\epsilon\searrow 0}\frac{\Lambda_{t+\epsilon}(\rho_0) - \Lambda_t (\rho_0)}{\epsilon} = \lim_{\epsilon\searrow 0}\frac{V_{t+\epsilon,t}-\id{}}{\epsilon}\circ \Lambda_t (\rho_0) = L_t (\rho_t),
\end{equation}
which comes via divisibility, $\Lambda_{t+\epsilon} = V_{t+\epsilon,t}\circ\Lambda_t$ and $\Lambda_0 = \id{}$. We therefore have
\begin{equation}
	L_t = \lim_{\epsilon\searrow 0}\frac{V_{t+\epsilon,t}-\id{}}{\epsilon}.
\end{equation}
Let us apply the D-divisibility condition, i.e.~put
\begin{equation}
	V_{t+\epsilon, t} = X_{t+\epsilon,t} + \theta \circ Y_{t+\epsilon,t}
\end{equation}
for some continuous functions $(t,s)\mapsto X_{t,s},Y_{t,s}\in\cpe{\matrd}$, $s\leqslant t$. Maps $X_{t,s}$ and $Y_{t,s}$, being completely positive, admit operator-sum representations
\begin{equation}
	X_{t,s}(\rho) = \sum_{j,k=1}^{d^2} x_{jk}(t,s) F_j \rho F_k, \qquad Y_{t,s}(\rho) = \sum_{j,k=1}^{d^2} y_{jk}(t,s) F_j \rho F_k,
\end{equation}
where $\rho\in\matrd$, for some matrices $[x_{jk}(t,s)],[y_{jk}(t,s)]\in\matr{d^2}^+$, also continuously depending on $(t,s)$. Similarly, the transposition map $\theta$ admits a representation \eqref{eq:ThetaRep}, i.e.
\begin{equation}
	\theta(\rho) = \rho^\transpose = \sum_{i=1}^{d^2} \theta_i F_i \rho F_i
\end{equation}
where $\theta_i$ are given in \eqref{eq:ThetaIdef}. Therefore, the expression for $L_t$, using composition rule \eqref{eq:FiFjCompRule} and properties \eqref{eq:XiProp}, is
\begin{align}
	&L_t (\rho ) = \lim_{\epsilon\searrow 0}\frac{1}{\epsilon}\left[X_{t+\epsilon, t}(\rho) + Y_{t+\epsilon,t}(\rho )^\transpose - \rho\right] \\
	&= \lim_{\epsilon\searrow 0}\frac{1}{\epsilon}\left[ \sum_{j,k=1}^{d^2} x_{jk}(t+\epsilon,t) F_j \rho F_k + \sum_{j,k,l=1}^{d^2} \theta_l y_{jk}(t+\epsilon,t) F_l F_j \rho F_k F_l - \rho \right] \nonumber \\
	&= \lim_{\epsilon\searrow 0}\frac{1}{\epsilon}\left[ \sum_{j,k=1}^{d^2} x_{jk}(t+\epsilon,t) F_j \rho F_k + \sum_{j,k,l=1}^{d^2}\sum_{\mu,\nu=1}^{d^2} \theta_l  y_{jk}(t+\epsilon,t) \xi_{lj\mu}\overline{\xi_{lk\nu}} F_\mu \rho F_\nu - \rho \right] \nonumber \\
	&= \lim_{\epsilon\searrow 0}\frac{1}{\epsilon}\left[ \sum_{j,k=1}^{d^2} x_{jk}(t+\epsilon,t) F_j \rho F_k + \sum_{\mu,\nu=1}^{d^2} \sum_{j,k=1}^{d^2}\Omega_{\mu\nu}^{jk}  y_{jk}(t+\epsilon,t) F_\mu \rho F_\nu - \rho \right].\nonumber
\end{align}
Let now
\begin{equation}
	z_{\mu\nu}(t,s) = \sum_{j,k=1}^{d^2}\Omega_{\mu\nu}^{jk}  y_{jk}(t,s).
\end{equation}
It is easy to check that matrix $[z_{\mu\nu}(t,s)]_{\mu\nu}\in\matr{d^2-1}$ is Hermitian for every $(t,s)$, however is not positive semidefinite in general. Next, we subtract from both summations terms with $\mu,\nu=d^2$ and obtain, by $F_{d^2} = \frac{1}{\sqrt{d}}I_d$,
\begin{equation}
	L_t (\rho ) = \lim_{\epsilon\searrow 0}\left[W_{t,\epsilon}\rho + E_{t,\epsilon}\rho + \rho E_{t,\epsilon}^{\hadj} + \sum_{\mu,\nu=1}^{d^2-1} w_{\mu\nu}(t+\epsilon, t) F_\mu \rho F_\nu \right],
\end{equation}
where we introduced
\begin{subequations}
	\begin{equation}\label{eq:wMatrixDec}
		w_{\mu\nu}(t,s) = x_{\mu \nu}(t, s) + z_{\mu \nu}(t, s),
	\end{equation}
	\begin{equation}
		W_{t,\epsilon} = \left[\frac{1}{d}w_{d^2 d^2}(t+\epsilon,t)- 1 \right] I_d,
	\end{equation}
	\begin{equation}
		E_{t,\epsilon} = \frac{1}{\sqrt{d}} \sum_{\mu=1}^{d^2-1} w_{\mu d^2}(t+\epsilon, t) F_\mu .
	\end{equation}
\end{subequations}
Now, we define new time-dependent coefficients $g_{\mu\nu}(t)$ by setting
\begin{subequations}
	\begin{equation}
		g_{d^2d^2}(t) = \lim_{\epsilon\searrow 0}\frac{1}{\epsilon}\left[\frac{1}{d}w_{d^2 d^2}(t+\epsilon,t)- 1 \right],
	\end{equation}
	\begin{equation}
		g_{\mu\nu}(t) = \lim_{\epsilon\searrow 0} \frac{1}{\epsilon} w_{\mu \nu}(t+\epsilon, t), \quad 1\leqslant\mu,\nu\leqslant d^2 -1,
	\end{equation}
\end{subequations}
where existence of all limits is assured by differentiability of $\Lambda_t$, so our expression for $L_t (\rho)$ becomes
\begin{align}
	L_t (\rho) &= g_{d^2 d^2}(t) \rho + E_t\rho + \rho E_{t}^{\hadj} + \sum_{\mu,\nu=1}^{d^2-1} g_{\mu\nu}(t) F_\mu \rho F_\nu
\end{align}
for $E_t = \frac{1}{\sqrt{d}}\sum_{\mu=1}^{d^2}g_{\mu d^2}(t) F_\mu$. Introducing two new matrices
\begin{subequations}
	\begin{equation}
		A_t = \frac{1}{2}\left(E_t + E_{t}^{\hadj}\right) + \frac{1}{2}\gamma_{d^2 d^2}(t) I_d ,
	\end{equation}
	\begin{equation}
		J_t = -\frac{1}{2i} \left( E_t - E_{t}^{\hadj}\right),
	\end{equation}
\end{subequations}
we obtain
\begin{equation}
	L_t (\rho) = -i \comm{J_t}{\rho} + \acomm{A_t}{\rho} + \sum_{\mu,\nu=1}^{d^2-1} g_{\mu\nu}(t)F_\mu \rho F_\nu.
\end{equation}
We demand $V_{t,s}$ to obey the \emph{trace preservation} condition, which means that $L_t$ must nullify the trace, $\tr{L_t (\rho)} = 0$ regardless of $\rho$. This applied to our expression yields, after some algebra involving cyclic property of trace,
\begin{equation}
	A_t = -\frac{1}{2}\sum_{\mu,\nu = 1}^{d^2-1} g_{\mu\nu}(t) F_\nu F_\mu .
\end{equation}
By inserting back we therefore end up with a form
\begin{equation}\label{eq:LtRho}
	L_t(\rho) = -i\comm{J_t}{\rho} + \sum_{\mu,\nu = 1}^{d^2-1} g_{\mu\nu}(t) \left( F_\mu \rho F_\nu - \frac{1}{2}\acomm{F_\nu F_\mu}{\rho}\right),
\end{equation}
which despite its visual resemblance is \emph{not} the standard form, since matrix $[g_{\mu\nu}(t)]_{\mu\nu}$ is not positive semi-definite in general. However, formula \eqref{eq:wMatrixDec} allows to split coefficients $g_{\mu\nu}(t)$ into a sum of expressions defined solely via either the CP or the coCP part of the propagator, namely
\begin{equation}
	g_{\mu\nu}(t) = \gamma_{\mu\nu}(t) + \eta_{\mu\nu}(t),
\end{equation}
where
\begin{equation}
	\gamma_{\mu\nu}(t) = \lim_{\epsilon\searrow 0}\frac{1}{\epsilon} x_{\mu\nu}(t+\epsilon, t), \qquad \eta_{\mu\nu}(t) = \lim_{\epsilon\searrow 0}\frac{1}{\epsilon} \sum_{j,k=1}^{d^2} \Omega_{\mu\nu}^{jk} y_{jk}(t+\epsilon, t).
\end{equation}
In similar fashion, we have $J_t = H_t + K_t$ where
\begin{equation}
	H_t = \frac{i}{2\sqrt{d}} \left( \sum_{\mu=1}^{d^2-1}\gamma_{\mu d^2}(t) F_\mu - \sum_{\mu=1}^{d^2-1} \overline{\gamma_{\mu d^2}(t)} F_\mu \right)
\end{equation}
and $K_t$ has an identical structure, with $\eta_{\mu d^2}(t)$ in place of $\gamma_{\mu d^2}(t)$. It is then evident that expression \eqref{eq:LtRho} may be rewritten as a sum of two maps, $L_t = M_t + N_t$, acting on $\rho$, where
\begin{equation}\label{eq:Nt}
	N_t(\rho) = -i\comm{K_t}{\rho} + \sum_{\mu,\nu = 1}^{d^2-1} \eta_{\mu\nu}(t) \left( F_\mu \rho F_\nu - \frac{1}{2}\acomm{F_\nu F_\mu}{\rho}\right)
\end{equation}
and $M_t$ is of the same structure, with $H_t$ replacing $K_t$ and $\gamma_{\mu\nu}(t)$ in place of $\eta_{\mu\nu}(t)$. By direct check, matrices $H_t$ and $K_t$ are Hermitian and complete positivity of map $X_{t,s}$ yields both matrices $[x_{\mu\nu}(t,s)]$ and $[\gamma_{\mu\nu}(t)]$ to be positive semidefinite, i.e.~map $M_t$ is in standard form. It remains to show that coefficients $\eta_{\mu\nu}(t)$ are as claimed. We have
\begin{equation}
	\eta_{\mu\nu}(t) =\sum_{j,k=1}^{d^2} \Omega_{\mu\nu}^{jk} \lim_{\epsilon\searrow 0}\frac{1}{\epsilon}  y_{jk}(t+\epsilon, t).
\end{equation}
As we show in lemma \ref{lemma:OmegaMatrixPositive} in the Appendix, the above limiting procedure under the summation defines a positive semidefinite matrix for all $t\in[t_1, t_2]$, i.e.~we have
\begin{equation}
	\omega_{jk}(t) = \lim_{\epsilon\searrow 0}\frac{1}{\epsilon}  y_{jk}(t+\epsilon, t), \quad [\omega_{jk}(t)]\in\matr{d^2}^{+}
\end{equation}
and $\eta_{\mu\nu}(t)$ admits a form \eqref{eq:EtaMatrixOmega}. This proves sufficiency. To show necessity, we start with re-expressing $N_t$, basing on expression \eqref{eq:Nt}, as
\begin{align}\label{eq:NtFormula}
	N_t (\rho) = &-i \comm{K_t}{\rho} + \left( \sum_{\alpha=1}^{d^2} C_{\alpha,t} \rho C_{\alpha,t}^{\hadj} \right)^{\transpose} - \frac{1}{2} \comm{D_t - D_{t}^{\hadj}}{\rho} \\
	&- \frac{1}{2} \sum_{\alpha=1}^{d^2}\acomm{C_{\alpha,t}^{\hadj}C_{\alpha,t}}{\rho} \nonumber
\end{align}
which is achieved by: (1) expressing $\eta_{\mu\nu}(t)$ via \eqref{eq:EtaMatrixOmega}, (2) expanding the geometric tensor $\Omega^{jk}_{\mu\nu}$ according to \eqref{eq:OmegaTensor}, (3) applying the operator-sum representation \eqref{eq:ThetaRep} of transposition map $\theta$, (4) expressing $[\omega_{jk}(t)]\in\matr{d^2}^{+}$ as $\omega_{jk}(t) = \sum_{\alpha} c_{j\alpha}(t) \overline{c_{k\alpha}(t)}$ for some new matrix $[c_{ij}(t)]$ and finally (5) substituting
\begin{equation}
	C_{\alpha,t} = \sum_{j=1}^{d^2} c_{\alpha j}(t) F_j, \qquad D_t = \sum_{j,k=1}^{d^2} \omega_{jk}(t) A_k F_j,
\end{equation}
for $A_k = \sum_{l} \theta_l \overline{\xi_{lkd^2}}F_l$ (see the derivation in section \ref{app:NtDerivation} in the Appendix). The matrix $D_t - D_{t}^{\hadj}$ is clearly skew-Hermitian, so it is of a form $D_t - D_{t}^{\hadj} = -i E_t$ for some Hermitian $E_t$. Now, recall $M_t$ was in standard form, so matrix $[\gamma_{\mu\nu}(t)]$ is positive semidefinite, i.e.~it may be cast into a form
\begin{equation}
	\gamma_{\mu\nu}(t) = \sum_{i=1}^{d^2-1} g_{\mu i}(t) \overline{g_{\nu i}(t)}
\end{equation}
for some matrix $[g_{ij}(t)]\in\matr{d^2-1}$. Then, by defining $G_{\alpha,t} = \sum_i g_{\alpha i}(t) F_i$ we can rewrite $M_t$ as
\begin{equation}
	M_t (\rho) = -i \comm{H_t}{\rho} + \sum_{\alpha} \left(G_{\alpha,t} \rho G_{\alpha,t}^{\hadj} - \frac{1}{2}\acomm{G_{\alpha,t}^{\hadj}G_{\alpha,t}}{\rho}\right),
\end{equation}
which is sometimes referred to as the \emph{second standard form} of a generator. All of this allows to rewrite expression for $L_t$ as
\begin{equation}
	L_t = \tilde{L}_{t}^{(0)} + \tilde{L}_{t}^{(1)} + \tilde{L}_{t}^{(2)}
\end{equation}
where individual parts $\tilde{L}_{t}^{(i)}$ are defined via
\begin{subequations}
\begin{equation}
	\tilde{L}_{t}^{(0)}(\rho) = -i \comm{H_t + K_t + E_t}{\rho} - \frac{1}{2}\sum_{\alpha} \acomm{G_{\alpha,t}^{\hadj}G_{\alpha,t} + C_{\alpha,t}^{\hadj}C_{\alpha,t}}{\rho},
\end{equation}
\begin{equation}\label{eq:Lt1}
	\tilde{L}_{t}^{(1)}(\rho) = \sum_{\alpha} G_{\alpha,t}\rho G_{\alpha,t}^{\hadj},
\end{equation}
\begin{equation}\label{eq:Lt2}
	\tilde{L}_{t}^{(2)}(\rho) = \left(\sum_{\alpha} C_{\alpha,t}\rho C_{\alpha,t}^{\hadj}\right)^{\transpose} = \theta \left( \sum_{\alpha} C_{\alpha,t}\rho C_{\alpha,t}^{\hadj} \right).
\end{equation}
\end{subequations}
Now, select an increasing sequence $(\tau_{j})_{j=0}^{n} \subset [s,t]$ of instants such that $\tau_0 = s$ and $\tau_n = t$. Then, we can express the propagator $V_{t,s}$ in a form
\begin{equation}\label{eq:VtsTSF}
	V_{t,s} = \lim_{\max{|\tau_{j+1}-\tau_j|\to 0}}{\prod_{j=n-1}^{0}e^{(\tau_{j+1}-\tau_j)L_{\tau_j}}},
\end{equation}
i.e.~we approximate the exact propagator by a composition of \emph{semigroups}; this is known as the \emph{time-splitting formula} \cite{Rivas2012}. Denote $\tau_{j+1}-\tau_j = \Delta_j$. Applying decomposition \eqref{eq:LtDecomposition} we have, by \emph{Lie-Trotter product formula},
\begin{equation}\label{eq:ExpDeltaLtau}
	e^{\Delta_j L_{\tau_j}} = \exp{\left(\Delta_j \sum_{k=0}^{2}\tilde{L}_{\tau_j}^{(k)}\right)} = \lim_{n\to\infty}\left( \prod_{k=0}^{2}\exp{\frac{\Delta_j}{n}\tilde{L}_{\tau_j}^{(k)}} \right)^{n}.
\end{equation}
We now have to specify properties of three maps $\exp{\frac{\Delta_j}{n}\tilde{L}_{t}^{(k)}}$ for $k = 0$, $1$ and $2$:

\begin{enumerate}
	\item Case $k = 0$. Let us define 
\begin{equation}
	W = \frac{\Delta_j}{n}\left[ -i(H_t+K_t + E_t) - \frac{1}{2}\sum_{\alpha=1}^{d^2}\left( G_{\alpha,t}^{\hadj}G_{\alpha,t} + C_{\alpha,t}^{\hadj}C_{\alpha,t} \right) \right]
\end{equation}
for fixed $t$, $j$ and a mapping $\xi \mapsto f_\xi \in \cpe{\matrd}$ by setting
\begin{equation}
	f_\xi(\rho) = e^{\xi W} \rho e^{\xi {W}^\hadj}, \qquad \rho \in \matrd .
\end{equation}
Then, by direct calculation one can easily check that we have
\begin{equation}
	\frac{d}{d\xi}f_\xi(\rho) = \frac{\Delta_j}{n}\tilde{L}_{t}^{(0)} (f_\xi(\rho)),
\end{equation}
i.e.~the identity
\begin{equation}\label{eq:fXi}
	f_\xi = \exp{\frac{\xi\Delta_j}{n}\tilde{L}_{t}^{(0)}}
\end{equation}
holds for all $\xi\in\reals$, i.e.~$\{f_\xi : \xi\in\reals\}$ is a group of completely positive maps. In particular, $\exp{\frac{\Delta_j}{n}\tilde{L}_{t}^{(0)}} = f_1$ is CP.
\item Case $k=1$. Note that $\tilde{L}^{(1)}_{t}$ defined in \eqref{eq:Lt1} is a CP map (being in its Kraus form). Therefore $\exp{\frac{\Delta_j}{n}\tilde{L}_{t}^{(1)}}$ is also CP due to proposition \ref{prop:ExpOfCPmap} (see appendix \ref{app:DecomposableMaps}).
\item Case $k=2$. Finally, $\tilde{L}^{(2)}_{t}$ given via \eqref{eq:Lt2} is clearly a coCP map. Then, by virtue of proposition \ref{prop:ExpOfcoCPmap} (appendix \ref{app:DecomposableMaps}), the remaining map $\exp{\frac{\Delta_j}{n}\tilde{L}_{t}^{(2)}}$ is decomposable.
\end{enumerate}
In the result, the map appearing under the limit in expression \eqref{eq:ExpDeltaLtau} is decomposable for every $n$ (as a composition); this shows $e^{\Delta_j L_{\tau_{j}}}$ is also decomposable, since it is a limit of a sequence of decomposable maps in closed cone $\dece{\matrd}$. This very same fact then shows that $V_{t,s}$ given in \eqref{eq:VtsTSF} is also decomposable. Finally, one checks by direct calculation that $L_t = M_t + N_t$ nullifies the trace, i.e.~$\tr{L_t (\rho)} = 0$. This yields that a family $\{e^{\tau L_t} : \tau \in\reals_+\}$ must be trace preserving for every choice of $t\in\reals_+$; in consequence, every map $e^{\Delta_j L_{\tau_j}}$ in decomposition \eqref{eq:VtsTSF} is also trace preserving and so is the whole propagator $V_{t,s}$. This concludes the proof.
\end{proof}

We furnish our result with the following equivalent statement. Recall that, as a finite dimensional vector space, $\matrd$ is isomorphic to its algebraic dual $\matrd^\prime$ with duality pairing expressed in terms of the \emph{trace},
\begin{equation}
	\matrd^\prime\times\matrd \ni (f,a) = \tr{b_f a},
\end{equation}
where a mapping $f \mapsto b_f \in \matrd$ is a bijection. Let $\phi$ be a linear map on $\matrd$. Then, there exists another linear map $\phi^\prime$ on $\matrd$ such that
\begin{equation}
	\tr{a\,\phi(b)} = \tr{\phi^\prime (a)b}, \quad a,b \in \matrd,
\end{equation}
which we call \emph{dual} to $\phi$ (with a slight abuse of terminology). We have:

\begin{theorem}
\label{thm:MainTheorem2}
Family $\{\Lambda_t : t\in\reals_+\}$ of linear maps on $\matrd$, subject to equation \eqref{eq:LambdaODE} in interval $[t_1, t_2]\subseteq\reals_+$, is D-divisible if and only if there exists a Hermitian matrix $S_t \in \matrd$ and map $\varphi_t \in \dece{\matrd}$ such that the generator $L_t$ admits the form
\begin{equation}
	L_t = -i \comm{S_t}{\cdot\,} + \varphi_t - \frac{1}{2}\acomm{\varphi_{t}^{\prime}(I)}{\cdot\,}.
\end{equation}
\end{theorem}

\begin{proof}
It suffices to set two CP maps,
\begin{equation}
	\phi_t (\rho) = \sum_{\alpha} G_{\alpha,t}\rho G_{\alpha,t}^{\hadj}, \quad \psi_t (\rho) = \sum_\alpha C_{\alpha,t}\rho C_{\alpha,t}^{\hadj},
\end{equation}
where we used the same notation as in the proof of theorem \ref{thm:MainTheorem}. Then one checks that both parts $M_t$ and $N_t$ of the generator may be conveniently re-expressed as
\begin{subequations}
	\begin{equation}
			M_t = -i \comm{H_t}{\cdot\,} + \phi_t - \frac{1}{2}\acomm{\phi_{t}^{\adj}(I)}{\cdot\,} ,
	\end{equation}
	\begin{equation}
			N_t = -i \comm{K_t + E_t}{\cdot\,} + \psi_t - \frac{1}{2}\acomm{\psi_{t}^{\adj}(I)}{\cdot\,}
	\end{equation}
\end{subequations}
and their sum can be shown with a simple algebra to be in the claimed form after defining a decomposable map $\varphi_t$ and Hermitian matrix $S_t$ via
\begin{equation}
	\varphi_t = \phi_t + \theta \circ \psi_t, \quad S_t = H_t + K_t + E_t
\end{equation}
and notifying $(\theta \circ \psi_t)^\adj (I) = \psi_{t}^{\adj}(I)$.
\end{proof}

In order to confirm validity of our results, we verified if families given by $L_t$ in proposed form were indeed decomposable. We checked for condition stated in theorem \ref{thm:DecCondition} by minimizing the functional $\rho \mapsto \tr{C_{e^{tL}}\rho}$ over a convex set $V_d \cap V_{d}^{\Gamma}$. This was achieved via a numerical and symbolic application of SDP optimization routines for a very wide range of different forms of $L_t$ in different dimensions and values of $t$. 

\subsection{Asymptotic complete positivity}

In general, decomposability properties of D-divisible dynamical maps turn out to be quite surprising, as we were able to check numerically. For instance, it may happen that $\Lambda_t$ suddenly becomes \emph{completely positive} despite the fact that the propagator $V_{t,s}$ remains \emph{truly} decomposable, i.e.~has a non-zero coCP part. Behavior of $\Lambda_t$ in this manner may be quite complex and ranges from being simply CP to even fluctuating between complete positivity and decomposability. Under particular circumstances, i.e.~under specific choice of the generator, an interesting phenomenon of $\Lambda_t$ is observed: namely, it is possible that initially $\Lambda_t$ is decomposable and then it \emph{switches} to being only CP and remains such as time progresses. This observation justifies a following definition of \emph{asymptotic} complete positivity of decomposable maps:

\begin{definition}
We will call a family $\{\Lambda_t : t\in\reals_+\}$ \emph{asymptotically CP} if there exists $t_0 > 0$ such that $\Lambda_t$ is CP and trace preserving for all $t \geqslant t_0$.
\end{definition}

In fact, asymptotic complete positivity is observed even in simplest semigroup case, as an example (see below) demonstrates, and is analyzed by examining the spectrum of Choi matrix $C_{\Lambda_t}$. Since $\Lambda_t$ is Hermiticity preserving, $C_{\Lambda_t}$ is Hermitian and therefore it suffices that $\spec{C_{\Lambda_t}}\subset\reals_+$ for $\Lambda_t$ to be CP, which in turn is guaranteed if the smallest eigenvalue $\lambda_{\mathrm{min}} (C_\varphi)$ is non-negative. Therefore one should be interested at least in finding some well-behaved and computable lower bounds for smallest eigenvalues. One such bound was specified by Wolkowicz and Styan in \cite[Theorem 2.1]{Wolkowicz1980}. Let $A \in \matr{n}$ be a matrix of real spectrum, $\spec{A} = \{\lambda_{i}(A) : 1 \leqslant i \leqslant n\}$, $\lambda_i (A) \in \reals$. Then, the smallest eigenvalue $\lambda_{\mathrm{min}} (A)$ satisfies inequality
\begin{equation}\label{eq:wolkowiczBounds}
	\mu_A - \nu_A \sqrt{n-1} \leqslant \lambda_{\mathrm{min}} (A) \leqslant \mu_A - \frac{\nu_A}{\sqrt{n-1}},
\end{equation}
for $\mu_A = \frac{1}{n}\tr{A}$ and $\nu_{A}^{2} = \frac{1}{n} \tr{(A^2)} - \mu_{A}^{2}$. This allows to formulate a following sufficient condition for complete positivity:

\begin{proposition}
\label{prop:CPcondition}
A trace preserving map $\varphi \in \dece{\matrd}$ is CP if
\begin{equation}
	\sum_{i,j=1}^{d} \hsnorm{\varphi (E_{ij})}^{2} \leqslant \frac{d^2}{d^2-1},
\end{equation}
where $\hsnorm{a} = \sqrt{\tr{a^\hadj a}}$ stands for the \emph{Hilbert-Schmidt norm} of $a\in\matrd$.
\end{proposition}

\begin{proof}
Clearly $\varphi\in\cpe{\matrd}$ if $\lambda_{\mathrm{min}} (C_\varphi)$ is non-negative. By a simple algebra involving trace preservation of $\varphi$ one checks that
\begin{equation}
	\tr{C_\varphi} = d, \qquad \tr{C_{\varphi}^{2}} = \sum_{i,j=1}^{d} \tr{\varphi(E_{ij})\varphi(E_{ji})} = \sum_{i,j=1}^{d} \hsnorm{\varphi (E_{ij})}^{2},
\end{equation}
since $E_{ij} = E_{ji}^{\hadj}$ and $\varphi$ is Hermiticity preserving. This allows to check that $\lambda_{\mathrm{min}} (C_\varphi)$ satisfies
\begin{equation}\label{eq:LambdaMinLowerBound}
	\lambda_{\mathrm{min}} (C_\varphi) \geqslant \frac{1}{d}\left[ 1 - \sqrt{(d^2 -1) \left(\sum_{i,j=1}^{d} \hsnorm{\varphi (E_{ij})}^{2}-1\right)} \right] ,
\end{equation}
which comes from \eqref{eq:wolkowiczBounds} after putting $A = C_{\varphi}$, $n = d^2$. Finally, demanding the above lower bound to be non-negative yields the claim. 
\end{proof}
A following criterion of asymptotic complete positivity arises:
\begin{theorem}
\label{thm:AsymptCPSemigroupCondition}
Let $\{\Lambda_t : t \in \reals_+\}$ be D-divisible trace preserving family. If it happens that
\begin{equation}
	\lim_{t\to\infty}\sum_{i,j=1}^{d} \hsnorm{\Lambda_{t} (E_{ij})}^{2} < \frac{d^2}{d^2-1},
\end{equation}
then the family is asymptotically CP.
\end{theorem}

\begin{proof}
Let $g(t) = \sum_{i,j=1}^{d} \hsnorm{\Lambda_t (E_{ij})}^{2}$. If indeed $\lim_{t\to\infty} g(t) < \frac{d^2}{d^2-1}$ then by definition of a limit there exists $t_0 \geqslant 0$ such that $g(t) < \frac{d^2}{d^2-1}$ for all $t > t_0$ and we have $\lambda_{\mathrm{min}}(C_{\Lambda_t}) \geqslant 0$, $\Lambda_t \in \cpe{\matrd}$ by proposition \ref{prop:CPcondition}, i.e.~a family is asymptotically CP.
\end{proof}

\subsection{Decomposable semigroups}
\label{sec:Semigroups}

Here we briefly remark on the semigroup case. It is immediate that by suppressing all time dependence in decomposition \eqref{eq:LtDecomposition} we obtain a general characterization of D-divisible trace preserving semigroups over $\matrd$, for any $d$. Clearly, a semigroup is D-divisible if and only if it is decomposable, so we have a following corollary of theorem \ref{thm:MainTheorem}:

\begin{theorem}
A semigroup $\{e^{tL} : t\in\reals_+\}$ is trace preserving and decomposable iff $L$ is of a form stated in Theorems \ref{thm:MainTheorem} and \ref{thm:MainTheorem2}, with all matrices time independent.
\end{theorem}

We note here that an equivalent formula for generator $L$ in semigroup case was derived by Franke in 1976 \cite{Franke1976}, however with methods different from ours and without explicit utilization of decomposability.

In some cases, the limit appearing in theorem \ref{thm:AsymptCPSemigroupCondition} may be computed exactly. For example, if $L$ is \emph{diagonalizable}, its value turns out to be determined by the biorthogonal system of eigenbasis and associated dual basis of $L$:

\begin{theorem}
\label{thm:AsymptCPSemigroupConditionSpectrum}
Let $L$ be diagonalizable, let $0 \in \spec{L}$ be of multiplicity 1 and let $\varepsilon\in\ker{L}$ be an associated eigenmatrix. Then,
\begin{equation}
	\lim_{t\to\infty}\sum_{i,j=1}^{d} \hsnorm{e^{tL} (E_{ij})}^{2} = (\hsnorm{\varepsilon}\hsnorm{\beta})^{2} \geqslant 1,
\end{equation}
where $\beta\in\matrd$ is an element of dual basis of $L$ such that $\hsiprod{\beta}{\varepsilon} = 1$.
\end{theorem}

\begin{proof}
Let again $g(t) = \sum_{i,j=1}^{d} \hsnorm{e^{tL} (E_{ij})}^{2}$ and assume that $L$ is diagonalizable, i.e.~that there exists a linearly independent set $\{\vec{e}_i\}$ spanning $\complexes^{d^2}$ of (not necessarily orthogonal) normalized eigenvectors of $\hat{L}\in\matr{d^2}$, the \emph{matrixized} version of $L$, as elaborated in section \ref{subsect:Vectorization}. Then, one can show that there always exists so-called \emph{dual basis} (or \emph{reciprocal basis}) $\{\vec{b}_i\}$, also spanning $\complexes^{d^2}$, which is subject to relation $\iprod{\vec{b}_i}{\vec{e}_j} = \delta_{ij}$, or that $(\{\vec{e}_i\},\{\vec{b}_i\})$ constitutes for a \emph{biorthogonal system}. Then, every operator $\hat{A}$ acting on $\complexes^{d^2}$ may be cast into a form
\begin{equation}
	\hat{A} = \sum_{i,j=1}^{d^2} a_{ij} \iprod{\vec{b}_i}{\cdot} \vec{e}_j
\end{equation}
for coefficients $a_{ij} = \iprod{\vec{b}_i}{\hat{A}\vec{e}_j}$. In particular, when basis $\{\vec{e}_i\}$ is chosen as an eigenbasis of $\hat{A}$, we have
\begin{equation}\label{eq:Apseudoeigendec}
	\hat{A} = \sum_{i=1}^{d^2} \lambda_i (A) \iprod{\vec{b}_i}{\cdot} \vec{e}_i ,
\end{equation}
where $\lambda_i (A) \in \spec{A}$, i.e.~$\hat{A}$ admits a pseudo-spectral decomposition as a combination of non-orthogonal rank one projection operators onto its eigenspaces. By diagonalizability, $\hat{A} = \hat{P}\hat{D}\hat{P}^{-1}$ for invertible matrix $\hat{P}$, built from eigenvectors $\vec{e}_i$ stacked column-by-column and diagonal matrix $\hat{D} = \operatorname{diag}{\{\lambda_i (A)\}}$. In result, every analytic function $f$ of $\hat{A}$ shares the same eigenspaces and $f(\hat{A}) = \hat{P} f(\hat{D})\hat{P}^{-1}$, i.e.~$\spec{f(\hat{A})} = \{f(\lambda_i (A))\}$. Let us therefore denote $\spec{\hat{L}} = \{\mu_i\}$ and assume $\hat{L}$ is diagonalizable. Then $\hat{L}$ admits a decomposition \eqref{eq:Apseudoeigendec} for eigenvalues $\mu_i$, eigenvectors $\vec{e}_i$ and associated dual vectors $\vec{b}_i$. Naturally, map $L$ itself is also diagonalizable and we have
\begin{equation}\label{eq:SemigroupDec}
	e^{tL} = \sum_{i=1}^{d^2} e^{\mu_i t} \hsiprod{\beta_i}{\cdot} \, \varepsilon_i ,
\end{equation}
where $\beta_i = \operatorname{vec}^{-1}{\vec{b}_i}$, $\varepsilon_i = \operatorname{vec}^{-1}{\vec{e}_i}$ are eigenmatrices of $L$ and $\hsiprod{\beta_i}{\varepsilon_j} = \delta_{ij}$.

From general theory of positive unital maps, we know that $\spec{e^{tL}}$ lays inside unit circle (being a trace norm contraction), contains 1 (as a result of trace preservation) and is closed with respect to complex conjugation, i.e.~$e^{\mu_i t}, \overline{e^{\mu_i t}} \in \spec{e^{tL}}$ (by Hermiticity preservation property) \cite{Szczygielski_2021}. This implies that $0 \in \spec{L}$ and $\spec{L}\setminus\{0\}$ consists of pairs $\{\mu_i, \overline{\mu_i} : \Re{\mu_i}<0\}$ and possibly some negative reals. Let us then set $\mu_{1} = 0$. We have
\begin{equation}
	e^{\mu_{1}t} = 1 \,\,\,\text{and}\,\,\, e^{\mu_i t} = e^{-|\Re{\mu_i}|t} e^{i \Im{\mu_i} t}, \quad 1 \leqslant i \leqslant d^2-1,
\end{equation}
where we write  $-|\Re{\mu_i}|$ to emphasize negativity of real parts. Decomposition \eqref{eq:SemigroupDec} allows to re-write expression for $g(t)$. First, one easily confirms that
\begin{align}
	\hsiprod{e^{tL}(E_{ij})}{e^{tL}(E_{ij})} &= \sum_{k,l=1}^{d^2} \overline{e^{\mu_k t}}e^{\mu_l t} \hsiprod{\hsiprod{\beta_k}{E_{ij}} \varepsilon_k}{\hsiprod{\beta_l}{E_{ij}} \varepsilon_l} \\
	&= \sum_{k,l=1}^{d^2} e^{(\overline{\mu_k}+\mu_l)t} (\beta_k)_{ji} \overline{(\beta_l)_{ji}} \hsiprod{\varepsilon_k}{\varepsilon_l} \nonumber
\end{align}
which comes from properties of inner product and property $\tr{E_{ij}[a_{ij}]} = a_{ji}$. Substituting this into formula for $g(t)$ we have
\begin{align}
	g(t) &= \sum_{k,l=1}^{d^2} e^{(\overline{\mu_k}+\mu_l)t} z_{kl}
\end{align}
for shorthand notation $z_{kl} = \hsiprod{\beta_l}{\beta_k} \hsiprod{\varepsilon_k}{\varepsilon_l}$. Applying properties of eigenvalues $\mu_i$ we recast this into
\begin{align}
	g(t) = z_{11} &+ \sum_{k=2}^{d^2} e^{-2 |\Re{\mu_k}| t} z_{kk} \\
	&+ 2 \sum_{k<l} e^{-|\Re{\mu_k}| t} e^{-|\Re{\mu_l}|t} \Re{\left[ e^{i\Im{(\mu_k-\mu_l})t} z_{kl} \right]} \nonumber .
\end{align}
Clearly, both sums vanish exponentially as $t \to\infty$, so
\begin{equation}
	\lim_{t\to\infty}g(t) = z_{11} = \hsnorm{\beta_1}^{2} \hsnorm{\varepsilon_1}^{2}.
\end{equation}
By Schwartz inequality, $1 = \hsiprod{\beta_1}{\varepsilon_1} \leqslant \hsnorm{\beta_1} \hsnorm{\varepsilon_1}$, so indeed $\lim_{t\to\infty}g(t) \geqslant 1$, as claimed.\end{proof}

\section{Examples}
\label{sec:Examples}

Here we present two simple examples of D-divisible trace preserving evolution as outlined in preceding sections in low dimensional matrix algebras. The first one concerns a decomposable semigroup over $\matr{2}$, whereas as the second one we explore a very basic case of time-dependent generator in $\matr{3}$. For simplicity and readability of obtained formulas we choose the appropriate generators in simplest possible way, e.g. by choosing matrix $[\omega_{jk}(t)]$ as a diagonal one or neglecting some parts of generator $L_t$ (such as commutator terms).

\subsection{Decomposable semigroup on algebra \texorpdfstring{$\matr{2}$}{of complex 2-by-2 matrices}}

As a first example, we examine a decomposable semigroup on algebra of complex square matrices of size 2. We set
\begin{equation}
	L(\rho) = N(\rho) = \frac{1}{2}\sum_{\mu,\nu = 1}^{3} \eta_{\mu\nu} \left( \sigma_\mu \rho \sigma_\nu - \frac{1}{2}\acomm{\sigma_\nu \sigma_\mu}{\rho} \right),
\end{equation}
where $\sigma_i$ is the usual basis of Pauli matrices, i.e.~we explicitly neglect the $M$ generator from decomposition \eqref{eq:LtDecomposition} and the commutator part of \eqref{eq:NtDecomposition}. The geometric tensor $\hat{\mathbf{\Omega}}$ may be then computed by applying \eqref{eq:OmegaTensorRep}; its only non-zero coefficients $\Omega^{jk}_{\mu\nu}$ read
\begin{subequations}
	\begin{align}
		&\Omega^{11}_{11} = \Omega^{11}_{22} = \Omega^{12}_{12} = \Omega^{13}_{13} = \Omega^{13}_{31} = \Omega^{21}_{21} = \Omega^{22}_{11} = \Omega^{22}_{22} = \Omega^{22}_{33} \\
		= \, &\Omega^{23}_{23} = \Omega^{31}_{13} = \Omega^{31}_{31} = \Omega^{32}_{32} = \Omega^{33}_{22} = \Omega^{33}_{33} = \Omega^{44}_{11} = \Omega^{44}_{33} = \frac{1}{2} \nonumber ,
	\end{align}
	\begin{align}
		\Omega^{11}_{33} = \Omega^{12}_{21} = \Omega^{21}_{12} = \Omega^{23}_{32} = \Omega^{32}_{23} = \Omega^{33}_{11} = \Omega^{44}_{22} = -\frac{1}{2},
	\end{align}
	\begin{align}
		\Omega^{14}_{23} = \Omega^{14}_{32} = \Omega^{24}_{31} = \Omega^{42}_{31} = \Omega^{43}_{12} = \Omega^{43}_{21} = \frac{i}{2},
	\end{align}
	\begin{align}
		\Omega^{24}_{13} = \Omega^{34}_{12} = \Omega^{34}_{21} = \Omega^{41}_{23} = \Omega^{41}_{32} = \Omega^{42}_{13} = -\frac{i}{2}.
	\end{align}
\end{subequations}
For demonstration purpose of this example we choose a diagonal matrix $[\omega_{jk}]$,
\begin{equation}
	[\omega_{jk}] = \operatorname{diag}{\{w_1, \, ... \, , \, w_4\}}, \quad w_i \geqslant 0.
\end{equation}
Matrix $[\eta_{\mu\nu}]$ also admits a diagonal form
\begin{equation}
	[\eta_{\mu\nu}] = \frac{1}{2}\operatorname{diag}{\{w_1+w_2-w_3+w_4 , w_1+w_2+w_3-w_4, -w_1 + w_2 + w_3 + w_4\}}.
\end{equation}
After some computations, one arrives at the generator $L$,
\begin{equation}
	L(\rho) = \frac{1}{2}\left( \begin{array}{cc} s_{12} (\rho_{11}-\rho_{22}) & -(s_{23}+s_{24})\rho_{12} + (w_4-w_3)\rho_{21} \\ -(s_{23}+s_{24})\rho_{21} + (w_4-w_3)\rho_{12} & s_{12} (\rho_{11}-\rho_{22}) \end{array}\right)
\end{equation}
for $s_{ij} = w_i + w_j$. Next, performing the vectorization of $L$ (which we omit here for brevity) we obtain its spectrum,
\begin{equation}
	\spec{L} = \{\mu_1 = 0, \mu_2 = -s_{12}, \mu_3 = -s_{23}, \mu_4 = -s_{24}\},
\end{equation}
as well as corresponding eigenmatrices $\varepsilon_i$ such that $L(\varepsilon_i) = \mu_i \varepsilon_i$, which in this particular case happen to be equivalent to Pauli matrices,
\begin{equation}
	\varepsilon_1 = \frac{1}{\sqrt{2}} I, \quad \varepsilon_2 = -\frac{1}{\sqrt{2}} \sigma_3, \quad \varepsilon_3 = \frac{1}{\sqrt{2}} \sigma_1, \quad \varepsilon_4 = -\frac{i}{\sqrt{2}} \sigma_2.
\end{equation}
In such case, a dual basis is identical, $\beta_i = \varepsilon_i$. Dynamical semigroup $e^{tL}$ can be then (again, by vectorization techniques) characterized by its action on matrix $\rho = [\rho_{ij}]$ via
\begin{equation}
	\rho_t = [\rho_{ij}(t)] = e^{tL}([\rho_{ij}]),
\end{equation}
for matrix elements
\begin{subequations}
	\begin{equation}
		\rho_{11}(t) = \frac{1}{2} \left(1 + e^{-s_{12}t}\right) \rho_{11} + \frac{1}{2} \left(1 - e^{-s_{12}t}\right) \rho_{22},
	\end{equation}
	\begin{equation}
		\rho_{21}(t) = \frac{1}{2} \left(e^{-s_{23}t}-e^{-s_{24}t}\right)\rho_{12} + \frac{1}{2} \left(e^{-s_{23}t}+e^{-s_{24}t}\right)\rho_{21},
	\end{equation}
	\begin{equation}
		\rho_{12}(t) = \frac{1}{2} \left(e^{-s_{23}t}-e^{-s_{24}t}\right)\rho_{21} + \frac{1}{2} \left(e^{-s_{23}t}+e^{-s_{24}t}\right)\rho_{12},
	\end{equation}
	\begin{equation}
		\rho_{22}(t) = \frac{1}{2} \left(1 + e^{-s_{12}t}\right) \rho_{22} + \frac{1}{2} \left(1 - e^{-s_{12}t}\right) \rho_{11}.
	\end{equation}
\end{subequations}
The Choi matrix $C_{e^{tL}}$ is Hermitian as expected and reads
\begin{equation}
	C_{e^{tL}} = \frac{1}{2}\left( \begin{array}{cccc} 1+e^{-s_{12}t} & 0 & 0 & e^{-s_{23}t}+e^{-s_{24}t} \\
	0 & 1-e^{-s_{12}t} & e^{-s_{23}t}-e^{-s_{24}t} & 0 \\
	0 & e^{-s_{23}t}-e^{-s_{24}t} & 1-e^{-s_{12}t} & 0 \\
	e^{-s_{23}t}+e^{-s_{24}t} & 0 & 0 & 1+e^{-s_{12}t}
	\end{array}\right)
\end{equation}
and its spectrum is found to be
\begin{subequations}
	\begin{equation}
		\lambda_1(C_{e^{tL}}) = \frac{1}{2} \left(1-e^{-s_{12}t}-e^{-s_{23}t}+e^{-s_{24}t}\right),
	\end{equation}
	\begin{equation}
		\lambda_2(C_{e^{tL}}) = \frac{1}{2} \left(1-e^{-s_{12}t}+e^{-s_{23}t}-e^{-s_{24}t}\right),
	\end{equation}
	\begin{equation}
		\lambda_3(C_{e^{tL}}) = \frac{1}{2} \left(1+e^{-s_{12}t}-e^{-s_{23}t}-e^{-s_{24}t}\right),
	\end{equation}
	\begin{equation}
		\lambda_4(C_{e^{tL}}) = \frac{1}{2} \left(1+e^{-s_{12}t}+e^{-s_{23}t}+e^{-s_{24}t}\right).
	\end{equation}
\end{subequations}
Depending on actual values of $w_i$, the smallest eigenvalue of $C_{e^{tL}}$ may change sign and monotonicity. It is then possible for the semigroup to exhibit a mixed behavior:
\begin{enumerate}
	\item it may be CP for all $t \geqslant 0$, when $\lambda_{\mathrm{min}} (C_{e^{tL}})$ is everywhere non-negative; exemplary plot regarding such situation is shown in fig. \ref{fig:AlwaysCP},
	
	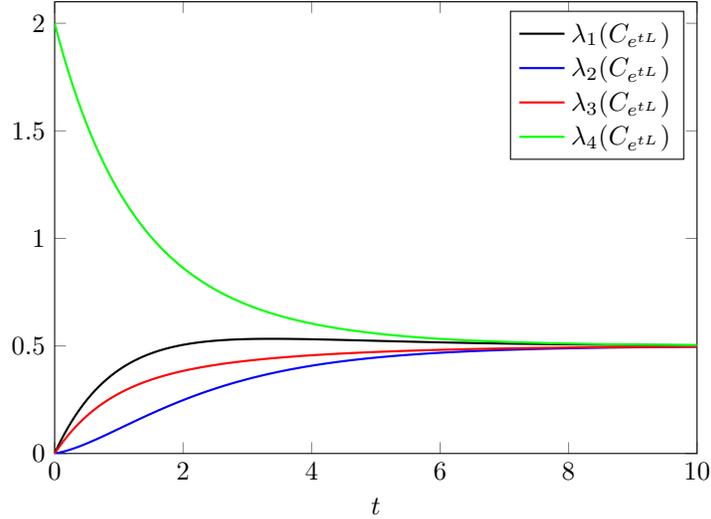
\begin{figure}[h!]
	\centering
	\begin{tikzpicture}
	\begin{axis}[
			xmin = 0, xmax = 10,
			ymin = 0, ymax = 2.1,
			width = 0.8\textwidth,
			height = 0.6\textwidth,
			xlabel = {$t$}
			]
			\addplot[
				black,
				domain = 0:10,
				samples = 200,
				smooth,
				thick,
			] {(1 - exp(-x * 0.7) - exp(-x * 1.1) + exp(-x * 0.5))/2};
			
			\addplot[
				blue,
				domain = 0:10,
				samples = 200,
				smooth,
				thick,
			] {(1 - exp(-x * 0.7) + exp(-x * 1.1) - exp(-x * 0.5))/2};
			
			\addplot[
				red,
				domain = 0:10,
				samples = 200,
				smooth,
				thick,
			] {(1 + exp(-x * 0.7) - exp(-x * 1.1) - exp(-x * 0.5))/2};

			\addplot[
				green,
				domain = 0:10,
				samples = 200,
				smooth,
				thick,
			] {(1 + exp(-x * 0.7) + exp(-x * 1.1) + exp(-x * 0.5))/2};
			
			\legend{
			{$\lambda_1 (C_{e^{tL}})$},
			{$\lambda_2 (C_{e^{tL}})$},
			{$\lambda_3 (C_{e^{tL}})$},
			{$\lambda_4 (C_{e^{tL}})$}
			}
	\end{axis}
	\end{tikzpicture}\caption{Spectrum of Choi matrix $C_{e^{tL}}$ as function of $t$ for parameters $w_1 = 0.3$, $w_2 = 0.4$, $w_3 = 0.7$, $w_4 = 0.1$. All eigenvalues are non-negative, i.e.~semigroup is always CP.}
	\label{fig:AlwaysCP}
	\end{figure}
	
	\item it may be decomposable (with both CP and coCP parts non-zero) for all $t \geqslant 0$, when $\lambda_{\mathrm{min}} (C_{e^{tL}}) < 0$ everywhere; see fig. \ref{fig:AlwaysDec},
	
	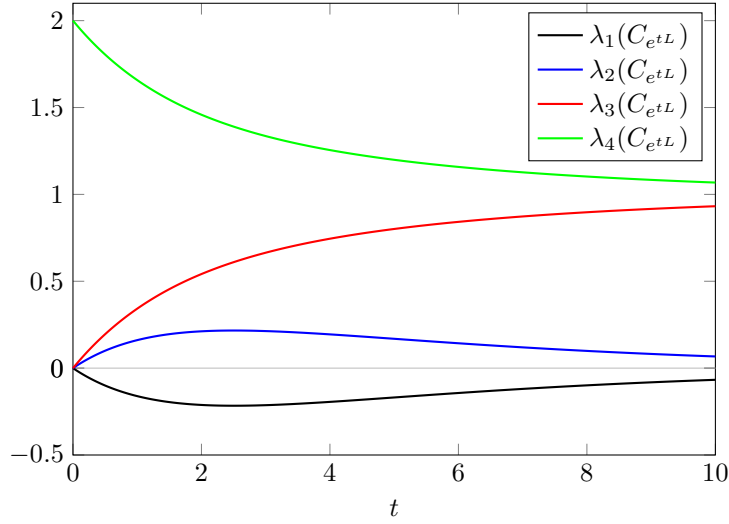
\begin{figure}[h!]
	\centering
	\begin{tikzpicture}
	\begin{axis}[
			xmin = 0, xmax = 10,
			ymin = -0.5, ymax = 2.1,
			width = 0.8\textwidth,
			height = 0.6\textwidth,
			xlabel = {$t$},
			extra y ticks = {0},
			extra tick style={grid=major},
			]
			\addplot[
				black,
				domain = 0:10,
				samples = 200,
				smooth,
				thick,
			] {(- exp(-x * 0.2) + exp(-x * 0.7))/2};
			
			\addplot[
				blue,
				domain = 0:10,
				samples = 200,
				smooth,
				thick,
			] {( exp(-x * 0.2) - exp(-x * 0.7))/2};
			
			\addplot[
				red,
				domain = 0:10,
				samples = 200,
				smooth,
				thick,
			] {(2 - exp(-x * 0.2) - exp(-x * 0.7))/2};

			\addplot[
				green,
				domain = 0:10,
				samples = 200,
				smooth,
				thick,
			] {(2 + exp(-x * 0.2) + exp(-x * 0.7))/2};
			
			\legend{
			{$\lambda_1 (C_{e^{tL}})$},
			{$\lambda_2 (C_{e^{tL}})$},
			{$\lambda_3 (C_{e^{tL}})$},
			{$\lambda_4 (C_{e^{tL}})$}
			}
	\end{axis}
	\end{tikzpicture}\caption{Spectrum of Choi matrix $C_{e^{tL}}$ as function of $t$ for parameters $w_1 = w_2 = 0$, $w_3 = 0.2$, $w_4 = 0.7$. One eigenvalue remains negative for all $t > 0$, i.e.~a semigroup is decomposable, yet never CP (except for $t=0$).}
	\label{fig:AlwaysDec}
	\end{figure}
	
	\item and finally, it can be decomposable in some interval $(0, t_0]$ and then become CP for $t \geqslant t_0$, i.e.~it may be asymptotically CP, as presented in fig.~\ref{fig:AsymCP}.
	
	\begin{figure}[h!]
	\centering
	\begin{tikzpicture}
	\begin{axis}[
			xmin = 0, xmax = 10,
			ymin = -0.3, ymax = 2.1,
			width = 0.8\textwidth,
			height = 0.6\textwidth,
			xlabel = {$t$},
			extra y ticks = {0},
			extra y tick style={grid=major},
			extra x ticks = {3.79},
			extra x tick label = {$t_0$},
			every extra x tick/.style={tick label style={fill=none, anchor=south east}}
			]
			\addplot[
				black,
				domain = 0:10,
				samples = 200,
				smooth,
				thick,
			] {(1 - exp(-x * 0.13) - exp(-x * 0.23) + exp(-x * 0.93))/2};
			
			\addplot[
				blue,
				domain = 0:10,
				samples = 200,
				smooth,
				thick,
			] {(1 - exp(-x * 0.13) + exp(-x * 0.23) - exp(-x * 0.93))/2};
			
			\addplot[
				red,
				domain = 0:10,
				samples = 200,
				smooth,
				thick,
			] {(1 + exp(-x * 0.13) - exp(-x * 0.23) - exp(-x * 0.93))/2};

			\addplot[
				green,
				domain = 0:10,
				samples = 200,
				smooth,
				thick,
			] {(1 + exp(-x * 0.13) + exp(-x * 0.23) + exp(-x * 0.93))/2};
			
			\addplot [dashed, latex-latex, samples=2, domain=-42.5:42.5] (3.79,x) node [pos=0.1, anchor=north, font=\footnotesize, sloped] {$x=3.79$};
			
			\legend{
			{$\lambda_1 (C_{e^{tL}})$},
			{$\lambda_2 (C_{e^{tL}})$},
			{$\lambda_3 (C_{e^{tL}})$},
			{$\lambda_4 (C_{e^{tL}})$}
			}
	\end{axis}
	\end{tikzpicture}\caption{Spectrum of Choi matrix $C_{e^{tL}}$ as function of $t$ for parameters $w_1 = 0.1$, $w_2 = 0.03$, $w_3 = 0.2$, $w_4 = 0.9$. The smallest eigenvalue $\lambda_{\mathrm{min}}(C_{e^{tL}})$ changes sign in neighborhood of $t_0 \approx 3.79$ and remains positive for all $t > t_0$ i.e.~a semigroup is asymptotically CP.}
	\label{fig:AsymCP}
	\end{figure}
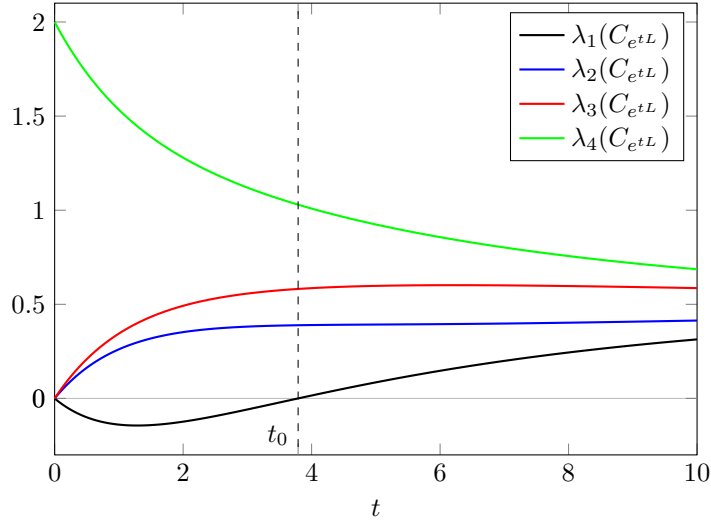
\end{enumerate}

\subsection{Time-dependent commutative Lindbladian in \texorpdfstring{$\mathbf{d=3}$}{d = 3}}
\label{sec:ExampleCommLind}

Our second example concerns a simple time-dependent Lindbladian over algebra $\matr{3}$ which we choose as
\begin{equation}
	L_t = g(t) L, \quad \text{where} \, \, g(t) = e^{-t}(1+\sin{\omega t}),
\end{equation}
and $L = \epsilon_1 M + \epsilon_2 N$ is constant and given as in \eqref{eq:LtDecomposition} and \eqref{eq:NtDecomposition}, however lacking commutator terms; $\epsilon_1, \epsilon_2, \omega \geqslant 0$ are dimensionless parameters. Just as earlier, we choose a diagonal matrix $[\omega_{jk}]$, this time of a form
\begin{equation}
	[\omega_{jk}] = \sum_{i=1}^{3} E_{ii} \otimes E_{ii} = \operatorname{diag}{\{1,\,0,\,0,\,0,\,1,\,0,\,0,\,0,\,1\}},
\end{equation}
which in result yields
\begin{equation}
	[\eta_{\mu\nu}] = \frac{1}{6} \operatorname{diag}{\{5,\,5,\,2,\,1,\,1,\,-2\}} \oplus \frac{1}{6} \left(\begin{array}{cc} -1 & \sqrt{3} \\ \sqrt{3} & 5 \end{array}\right).
\end{equation}
Matrix $[\gamma_{\mu\nu}]$ which defines the part $M$ of the generator is simply chosen to be identity, $\gamma_{\mu\nu} = \delta_{\mu\nu}$. Note, that function $f$ is always non-negative so positive semidefiniteness of matrices $[\gamma_{\mu\nu}]$ and $[\omega_{jk}]$ cannot be spoiled. The Hilbert-Schmidt orthonormal basis $\{F_{i}\}_{i=1}^{9}$ spanning $\matr{3}$ consists of Gell-Mann matrices (up to normalization); see Appendix \ref{subsect:HSbasis} for details. After evaluations, we obtain the action of maps $M$ and $N$,
\begin{subequations}
	\begin{align}\label{eq:Example2M}
		M(\rho) &= \sum_{\mu=1}^{8} \left( F_\mu \rho F_\mu - \frac{1}{2}\acomm{F_{\mu}^{2}}{\rho}\right) \\
		&= \left( \begin{array}{ccc} -2 \rho_{11} + \rho_{22} + \rho_{33} & -3 \rho_{12} & -3\rho_{13} \\ -3 \rho_{21} & \rho_{11} - 2\rho_{22} + \rho_{33} & -3\rho_{23} \\ -3 \rho_{31} & -3\rho_{32} & \rho_{11} + \rho_{22} -2\rho_{33} \end{array}\right), \nonumber
	\end{align}
	\begin{align}
		N(\rho) &= \sum_{\mu,\nu=1}^{8} \eta_{\mu\nu}\left( F_\mu \rho F_\mu - \frac{1}{2}\acomm{F_{\mu}^{2}}{\rho}\right) \\
		&= \frac{1}{12}\left( \begin{array}{ccc} 6(-2\rho_{11} + \rho_{22}+\rho_{33}) & 4\rho_{21}-7\rho_{12} & 4\rho_{31}-19\rho_{13} \\ 4\rho_{12}-7\rho_{21} & 6(\rho_{11}-\rho_{22}) & 2(2\rho_{32}-5\rho_{23}) \\ 4\rho_{13}-19\rho_{31} & 2( 2\rho_{23} - 5\rho_{32}) & 6(\rho_{11}-\rho_{33}) \end{array}\right). \nonumber
	\end{align}
\end{subequations}
Notice that generator $L_t$ satisfies commutativity condition $\comm{L_t}{L_s} = 0$ for any two chosen $t,s \in \reals_+$. This convenient property implies a particularly simple, formal expression for $\Lambda_t$,
\begin{equation}
	\Lambda_t = \exp{\int\limits_{0}^{t} L_{t^\prime} dt^\prime} = e^{f(t) L},
\end{equation}
where 
\begin{equation}
	f(t) = \int\limits_{0}^{t} g(t^\prime) dt^\prime = 1 - e^{-t} + \frac{1}{1+\omega^2}(\omega - \omega e^{-t}\cos{\omega t} - e^{-t}\sin{\omega t}).
\end{equation}
Map $\Lambda_t$ is then defined by its action, $\rho_t = \Lambda_t([\rho_{ij}])$, for explicit matrix elements
\begin{subequations}
	\begin{equation}
		\rho_{11}(t) = p_1 (t) \rho_{11} + p_2 (t) \rho_{22} + p_2(t)\rho_{33},
	\end{equation}
	\begin{equation}
		\rho_{22}(t) = p_2 (t) \rho_{11} + s_1 (t)\rho_{22} + s_2 (t) \rho_{33},
	\end{equation}
	\begin{equation}
		\rho_{33}(t) = p_2 (t) \rho_{11} + s_2 (t) \rho_{22} + s_1 (t) \rho_{33},
	\end{equation}
	\begin{equation}
		\rho_{21}(t) = q_1 (t) \rho_{12} + q_2 (t) \rho_{21}, \quad \rho_{12}(t) = q_2 (t) \rho_{12} + q_1 (t) \rho_{21},
	\end{equation}
	\begin{equation}
		\rho_{31}(t) = r_1 (t)\rho_{13} + r_2 (t) \rho_{31}, \quad \rho_{13}(t) = r_2 (t)\rho_{13} + r_1 (t) \rho_{31},
	\end{equation}
	\begin{equation}
		\rho_{32}(t) = u_1 (t)\rho_{23} + u_2 (t)\rho_{32}, \quad \rho_{23}(t) = u_1 (t)\rho_{32} + u_2 (t)\rho_{23},
	\end{equation}
\end{subequations}
and functions
\begin{subequations}
	\begin{equation}
		p_1 (t) = \frac{1}{3} \left( 1 + 2e^{-\frac{3}{2}(2\epsilon_1 + \epsilon_2) f(t)} \right),
	\end{equation}
	\begin{equation}
		p_2 (t) = \frac{1}{3} \left( 1 - e^{-\frac{3}{2}(2\epsilon_1 + \epsilon_2) f(t)} \right),
	\end{equation}
	\begin{equation}
		q_{1,2} (t) = \frac{1}{2} e^{-\frac{1}{4} (12 \epsilon_1+\epsilon_2) f(t)}\left(1\mp e^{-\frac{2}{3} \epsilon_2 f(t)}\right) ,
	\end{equation}
	\begin{equation}
		r_1 (t) = \frac{1}{2} e^{-(3 \epsilon_1+\frac{23}{12}\epsilon_2) f(t)}\left(-1 + e^{-\frac{2}{3} \epsilon_2 f(t)}\right),
	\end{equation}
	\begin{equation}
		r_2 (t) = \frac{1}{2} e^{-(3 \epsilon_1+\frac{5}{4}\epsilon_2) f(t)}\left(1 + e^{-\frac{2}{3} \epsilon_2 f(t)}\right),
	\end{equation}
	\begin{equation}
		s_{1,2} (t) = \frac{1}{6} \left[2+e^{-(3 \epsilon_1+\frac{3}{2} \epsilon_2) f(t)}\left(1\pm 3 e^{\epsilon_2 f(t)}\right) \right],
	\end{equation}
	\begin{equation}
		u_{1,2}(t) = \frac{1}{2}  e^{-(3 \epsilon_1 +\frac{1}{2}\epsilon_2) f(t)} \left(1\mp e^{-\frac{2}{3} \epsilon_2 f(t)}\right).
	\end{equation}
\end{subequations}
After some effort, one can calculate the associated Choi matrix $C_{\Lambda_t}$ and its spectrum (for sake of reader's convenience we chose to avoid presenting the resulting cumbersome formulas), at least numerically for chosen values of parameters. Similar to the previous semigroup example, we had examined the time dependence of $\lambda_{\mathrm{min}}(C_{\Lambda_t})$, the smallest eigenvalue of Choi matrix, for a wide range of $\epsilon_1$, $\epsilon_2$ and $\omega$ and found the behavior of $\Lambda_t$ to be in parallel with the semigroup case, i.e.~$\Lambda_t$ may be always CP (when $\lambda_{\mathrm{min}}(C_{\Lambda_t}) \geqslant 0$, $t \geqslant 0$), always decomposable (i.e.~with coCP part non-zero, when $\lambda_{\mathrm{min}}(C_{\Lambda_t}) < 0$, $t\geqslant 0$) or asymptotically CP (when $\lambda_{\mathrm{min}}(C_{\Lambda_t}) \geqslant 0$ for all $t \geqslant t_0$), depending on parameters $\epsilon_{1,2}$. Some exemplary plots of $\lambda_{\mathrm{min}}(C_{\Lambda_t})$ are presented in fig.~\ref{fig:Example2}. Clearly, lowering the $\epsilon_1/\epsilon_2$ ratio decreases the significance of part $M$ \eqref{eq:Example2M} of the generator and pushes the dynamics from global complete positivity towards decomposability.

\begin{figure}[h!]
\centering
\begin{tikzpicture}
\begin{axis}[
		xmin = 0, xmax = 5,
		ymin = -0.05, ymax = 0.06,
		width = 0.8\textwidth,
		height = 0.6\textwidth,
		xlabel = {$t$},
		extra y ticks = {0},
		extra y tick style={grid=major},
		yticklabel style={
        /pgf/number format/fixed,
        /pgf/number format/precision=2
				},
		scaled y ticks=false,
		every extra x tick/.style={tick label style={fill=none, anchor=south east}},
		legend style={nodes={scale=0.7, transform shape}}, legend image post style={mark=*},
		legend cell align={left}
		]
		
\addplot[color=red] graphics[xmin=0,xmax=5,ymin=-0.05,ymax=0.06] {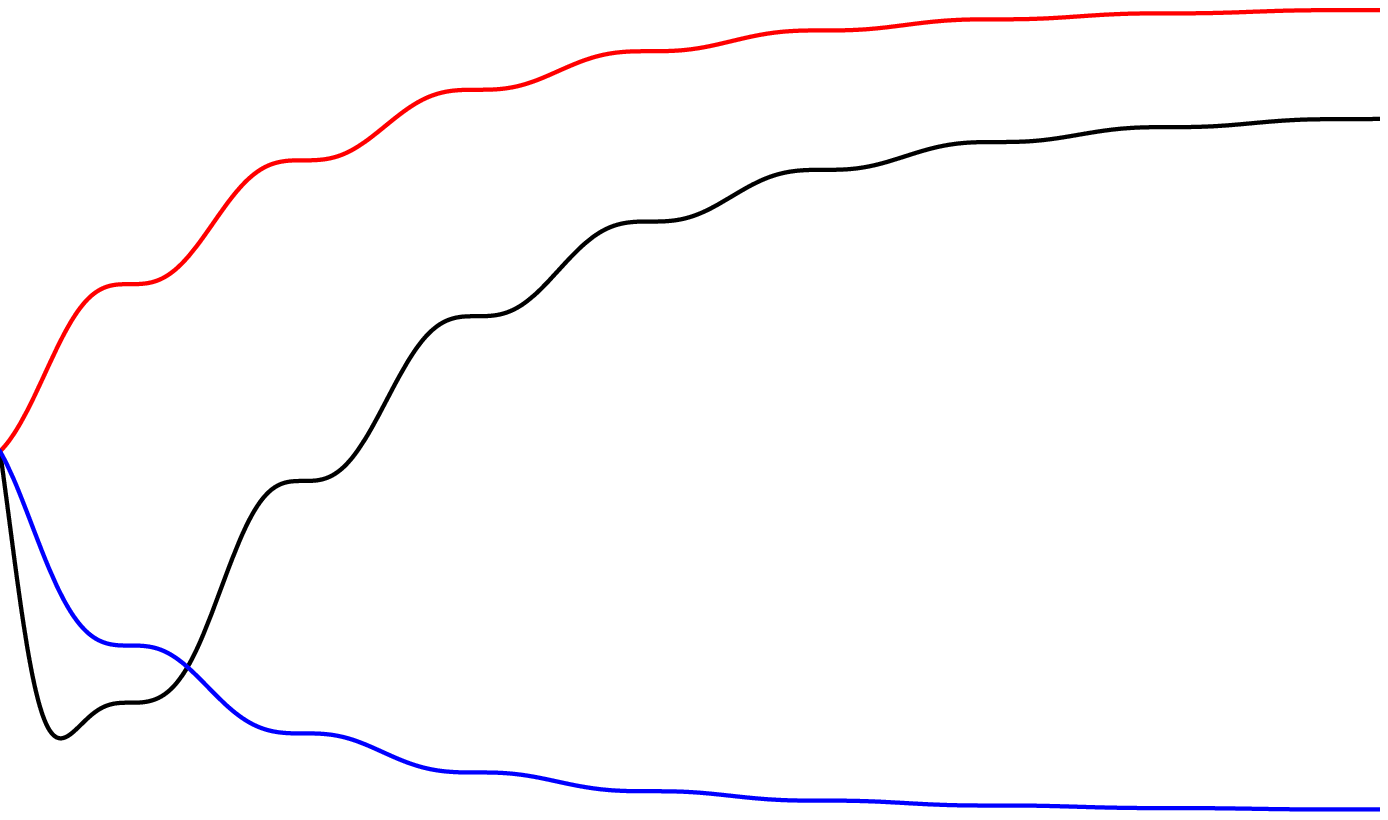};
\addplot[color=blue, draw=none] {0};
\addplot[color=black, draw=none] {0};

\legend{
{$\epsilon_1 = 0.1$, $\epsilon_2 = 0.2$ (CP)},
{$\epsilon_1 = 0.01$, $\epsilon_2 = 0.2$ (decomp.)},
{$\epsilon_1 = 0.1$, $\epsilon_2 = 1.0$ (asymp. CP)}
}
\end{axis}
\end{tikzpicture}
\caption{Time dependence of $\lambda_{\mathrm{min}}(C_{\Lambda_t})$ (numerically obtained) in example \ref{sec:ExampleCommLind} for different values of $\epsilon_1$, $\epsilon_2$ and fixed $\omega = 10$, showing three possible regimes of the dynamical map $\Lambda_t$ being CP ($\epsilon_1 = 0.1$, $\epsilon_2 = 0.2$), decomposable ($\epsilon_1 = 0.01$, $\epsilon_2 = 0.2$) and asymptotically CP ($\epsilon_1 = 0.1$, $\epsilon_2 = 1.0$).}
\label{fig:Example2}
\end{figure}

\section{Acknowledgments}

The author is indebted to Prof.~Dariusz Chru\'{s}ci\'{n}ski for discussion and to anonymous Referee for comments and suggestions which led to certain improvements in the \emph{Examples} section.

\section{Data availability}

No new data were created or analysed in this study.

\appendix

\section{Mathematical supplement}

\subsection{Hermitian Hilbert-Schmidt basis}
\label{subsect:HSbasis}

Let $\{F_{i}\}_{i=1}^{d^2}$ be the Hermitian Hilbert-Schmidt orthonormal basis in $\matrd$ subject to conditions \eqref{eq:HSonb}, i.e.
\begin{equation}
	F_{j} = F_{j}^{\hadj}, \quad \tr{F_{j}F_{k}} = \delta_{jk}, \quad \tr{F_j} = \delta_{j d^2}, \quad F_{d^2} = \frac{1}{\sqrt{d}} I.
\end{equation}
Matrices $\{F_i\}$ can be then constructed explicitly in a following way \cite{Hioe1981,Kimura2003}. Let again $E_{jk}$ denote matrix units, i.e.~they contain 1 in position $(j,k)$ and 0s elsewhere. Let us define matrices $W_{kj}^{d}, K_{k}^{d} \in \matrd$ such that
\begin{equation}
	W_{jk} = \begin{cases} \frac{1}{\sqrt{2}} \left(E_{jk}+E_{kj}\right), \quad \text{for } k<j, \\ -\frac{i}{\sqrt{2}}\left(E_{jk}-E_{kj}\right), \quad \text{for } k>j, \end{cases}
\end{equation}
such that $j,k \in \{1, \, ... \, , \, d^2 - 1\}$, $j\neq k$, as well as
\begin{equation}
	K_{k} = \frac{1}{\sqrt{k(k+1)}} \left( \sum_{j=1}^{k} E_{jj} - k E_{k+1,k+1} \right),
\end{equation}
where $k \in \{1, \, ... \, , \, d - 1\}$. Then, the set $\{W_{jk}, \, K_{k}, \, \frac{1}{\sqrt{d}} I_d\}$ contains $d^2$ matrices and is orthonormal (with respect to Hilbert-Schmidt inner product) and complete, being a basis of $\matrd$. Its elements are then labeled $F_i$ for $1\leqslant i \leqslant d^2$. Matrices $W_{jk}$ are either \emph{symmetric off-diagonal} or \emph{antisymmetric} and matrices $K_k$ are \emph{diagonal and of zero trace}. By simple counting, there is then exactly $\frac{1}{2}d(d-1)$ of both symmetric off-diagonal and antisymmetric matrices and $d$ diagonal matrices (including $F_{d^2} = \frac{1}{\sqrt{d}}I$).

One then introduces the so-called \emph{structure constants} $f_{ijk}$ and $g_{ijk}$, which respectively define the commutation and anticommutation relations amongst matrices $F_i$,
\begin{equation}
	\comm{F_i}{F_j} = \sum_{k=1}^{d^2-1} f_{ijk} F_k, \quad \acomm{F_i}{F_j} = \sum_{k=1}^{d^2-1} g_{ijk} F_k ,
\end{equation}
being defined as
\begin{equation}
	f_{ijk} = \tr{F_{k}\comm{F_i}{F_j}}, \quad g_{ijk} = \tr{F_{k}\acomm{F_i}{F_j}}.
\end{equation}
It is worth noting that structure constants characterize $\matrd$ as a Lie algebra. These allow us to derive a following composition rule
\begin{equation}
	F_a F_b = \sum_{c=1}^{d^2} \xi_{abc} F_{c}
\end{equation}
for coefficients $\xi_{ijk} = \frac{1}{2}(f_{ijk} + g_{ijk}) = \tr{F_i F_j F_k}$.

\subsubsection{Cases \texorpdfstring{$d = 2, 3$}{d = 2, 3}}

When $d = 2$, matrices $F_i$ are proportional to usual Pauli matrices:
\begin{equation}
	F_1 = \frac{1}{\sqrt{2}}\left(\begin{array}{cc} 0 & 1 \\ 1 & 0 \end{array}\right), \quad F_2 = \frac{1}{\sqrt{2}}\left(\begin{array}{cc} 0 & -i \\ i & 0 \end{array}\right), \quad F_3 = \frac{1}{\sqrt{2}}\left(\begin{array}{cc} 1 & 0 \\ 0 & -1 \end{array}\right),
\end{equation}
and $F_4 = \frac{1}{\sqrt{2}} I$. For $d=3$ instead, resulting matrices take the form
\begin{align}
	&F_1 = \frac{1}{\sqrt{2}}\left(\begin{array}{ccc} 0 & 1 & 0 \\ 1 & 0 & 0 \\ 0 & 0 & 0 \end{array}\right), \quad F_2 = \frac{1}{\sqrt{2}}\left(\begin{array}{ccc} 0 & 0 & 1 \\ 0 & 0 & 0 \\ 1 & 0 & 0 \end{array}\right), \\
	&F_3 = \frac{1}{\sqrt{2}}\left(\begin{array}{ccc} 0 & 0 & 0 \\ 0 & 0 & 1 \\ 0 & 1 & 0 \end{array}\right), \quad F_4 = \frac{1}{\sqrt{2}}\left(\begin{array}{ccc} 0 & -i & 0 \\ i & 0 & 0 \\ 0 & 0 & 0 \end{array}\right), \nonumber \\
	&F_5 = \frac{1}{\sqrt{2}}\left(\begin{array}{ccc} 0 & 0 & -i \\ 0 & 0 & 0 \\ i & 0 & 0 \end{array}\right), \quad F_6 = \frac{1}{\sqrt{2}}\left(\begin{array}{ccc} 0 & 0 & 0 \\ 0 & 0 & -i \\ 0 & i & 0 \end{array}\right), \nonumber \\
	&F_7 = \frac{1}{\sqrt{2}}\left(\begin{array}{ccc} 1 & 0 & 0 \\ 0 & -1 & 0 \\ 0 & 0 & 0 \end{array}\right), \quad F_8 = \frac{1}{\sqrt{6}}\left(\begin{array}{ccc} 1 & 0 & 0 \\ 0 & 1 & 0 \\ 0 & 0 & -2 \end{array}\right) \nonumber
\end{align}
and $F_9 = \frac{1}{\sqrt{3}}I$, i.e.~they are proportional to Gell-Mann matrices.

\subsection{Vectorization and matrixization}
\label{subsect:Vectorization}

Recall, that $\matrd$ is isomorphically identified with $\complexes^{d^2}$ and $B(\matrd)$ with $\matr{d^2}$. It is then very common and convenient to utilize these identifications in order to represent matrices as (column) vectors and linear maps on $\matrd$ as matrices of size $d^2$.

Every bijection $\matrd \to \complexes^{d^2}$ defines so-called \emph{vectorization} scheme \cite{Miszczak2011,Bengtsson2017}. A convenient vectorization, which we here denote by $\mathbf{vec}$, is the one given as the operation of \emph{flattening} of a matrix -- namely, for a matrix $[m_{ij}] \in \matrd$ we define a unique vector $\vec{m} = \operatorname{\mathbf{vec}}{([m_{ij}])}\in\complexes^{d^2}$ via \cite{Bengtsson2017}
\begin{equation}\label{eq:TheVec}
	\vec{m} = \operatorname{\textbf{vec}}{([m_{ij}])} = ( m_{11}, \, m_{12}, \, ... \, ,\, m_{1d}, \, m_{21}, \, m_{22}, \, ...\, , \, m_{dd} )^{\transpose},
\end{equation}
i.e.~by putting rows of $[m_{ij}]$ one behind another, or in a \emph{lexicographic order}. We remark here, that the convention of vectorization we use in this article is by no means universal. For example, some authors prefer the matrix flattening not in a row-by-row manner, but rather in column-by-column manner, which is sometimes called a \emph{reshaping}. For details, see \cite{Miszczak2011} and references within. The inverse operation $\mathbf{vec}^{-1} : \complexes^{d^2} \to \matrd$ reforms vectors back into matrices by splitting them into $d$-tuples and stacking one behind the other; such operation is sometimes called \emph{matrixization}. Every linear map $T$ on $\matrd$ then admits a unique representation as a matrix $\hat{T}\in\matr{d^2}$ in such a way, that for any $m\in\matrd$, matrix $T(m)$ is identified with $\hat{T}\vec{m}$, i.e.~$T(m) = \operatorname{\mathbf{vec}^{-1}}{(\hat{T}\vec{m})}$.

\subsection{Linear maps on matrix algebra}

\subsubsection{Operator-sum representation}
\label{subsect:OperatorSumRep}

Let $T : \matr{n} \to \matr{m}$ be linear. Then, there exist two nonunique, finite families of matrices $\{A_i\},\{B_i\}\in\matr{m, n}$ such that action of $T$ on any $a\in\matr{n}$ can be expressed as
\begin{equation}\label{eq:OSR}
	T(a) = \sum_{i} A_i a B_{i}^{\hadj},
\end{equation}
where it is customary to put the Hermitian conjugation of matrix $B_i$. Form \eqref{eq:OSR} is called the \emph{operator-sum representation} of $T$. For example, any Hermiticity preserving map possesses a form
\begin{equation}
	T(a) = \sum_{i} \lambda_i A_i a A_{i}^{\hadj}
\end{equation}
for some family of matrices $\{A_i\}$ and real coefficients $\lambda_i$ \cite{Pillis1967}. If in addition all $\lambda_i \geqslant 0$, then $T$ is completely positive.

Assume $T$ is an endomorphism over $\matrd$. Expanding matrices $A_i$, $B_i$ in basis $\{F_i\}$ one quickly checks that \eqref{eq:OSR} can be equivalently expressed as
\begin{equation}\label{eq:OSRF}
	T(a) = \sum_{i,j=1}^{d^2} t_{ij} F_i a F_j
\end{equation}
for some coefficients $t_{ij}\in\complexes$. Then, we easily see that $T$ is Hermiticity preserving if and only if $[t_{ij}]$ is Hermitian and CP if and only if $[t_{ij}]\geqslant 0$. We have a following

\begin{proposition}
\label{prop:GeneralMapF}
Matrix $[t_{ij}]\in\matr{d^2}$ in decomposition \eqref{eq:OSRF} may be computed as
\begin{equation}\label{eq:tCoeffFormula}
	t_{ij} = \tr{\left[(F_i \otimes \overline{F_k})^{\hadj} \hat{T}\right]},
\end{equation}
where $\hat{T} \in \matr{d^2}$ is a \emph{matricial representation} of $T$ under vectorization scheme elaborated in section \ref{subsect:Vectorization}.
\end{proposition}

\begin{proof}
It may be shown \cite{Miszczak2011,Bengtsson2017} that the mapping $a \mapsto AaB$, for $a,A,B\in\matrd$, can be represented under the vectorization scheme \eqref{eq:TheVec} as a matrix $A \otimes B^\transpose$ where $\otimes$ is the usual Kronecker product of matrices, i.e.
\begin{equation}
	\operatorname{\mathbf{vec}}{(AaB)} = (A\otimes B^\transpose) \vec{a}.
\end{equation}
This means that general prescription for linear map \eqref{eq:OSRF} is equivalently represented as a matrix $\hat{T}$ of size $d^2$ of a form
\begin{equation}
	\hat{T}=\sum_{i,j=1}^{d^2} t_{ij} (F_i \otimes F_{j}^{\transpose}).
\end{equation}
Notice that $\{F_{i}^{\transpose}\}$ is still a Hilbert-Schmidt orthonormal basis in space $\matrd^\transpose \simeq \matrd$, and so a set $\{F_i \otimes F_{j}^{\transpose}\}$ spans space $\matrd\otimes\matrd^\transpose \simeq \matr{d^2}$ being still a Hermitian Hilbert-Schmidt basis. This, together with Hermiticity of $F_i$ immediately implies
\begin{equation}
	t_{ij} = \hsiprod{F_i \otimes F_{j}^{\transpose}}{\hat{T}} = \tr{\left[(F_{i}\otimes\overline{F_j})^\hadj \hat{T}\right]},
\end{equation}
which is the claim.
\end{proof}

\subsubsection{Transposition map}

We grant a special attention to a transposition map, i.e.~a linear, Hermiticity and trace preserving map $\theta : \matrd \to \matrd$ acting via prescription $\theta ([a_{ij}]) = [a_{ji}]$. Let again a space $\matr{d}$ be spanned by a Hilbert-Schmidt orthonormal basis $\{F_{i}\}$ satisfying properties \eqref{eq:HSonb}. Then we have a following result:

\begin{proposition}
Let
\begin{equation}\label{eq:ThetaMatrix}
	\hat{\theta} = \operatorname{diag}{\{\theta_1, \, ... \, , \, \theta_{d^2}\}} = I_{\frac{1}{2}d(d-1)} \oplus \left(-I_{\frac{1}{2}d(d-1)}\right) \oplus I_d ,
\end{equation}
Define also a set
\begin{equation}
	\mathcal{J} = \{1+\frac{1}{2}d(d-1), \, ... \, , \, d(d-1)\}.
\end{equation}
Then, the transposition map $\theta$ admits an operator-sum representation of a form
\begin{equation}
	\theta (a) = a^{\transpose} = \sum_{i=1}^{d^2} \theta_i F_i a F_{i}
\end{equation}
for coefficients $\theta_i \in \{-1, \, 1\}$ given explicitly as
\begin{equation}
	\theta_i = \begin{cases} -1, \quad \text{for } i\in\mathcal{J}, \\ +1, \quad \text{otherwise,} \end{cases}
\end{equation}
which therefore yields
\begin{equation}\label{eq:ThetaAlternateForm}
	\theta (a) = \sum_{i=1}^{d^2} F_i a F_i - 2 \sum_{i\in \mathcal{J}} F_i a F_i .
\end{equation}
\end{proposition}

\begin{proof}
From proposition \ref{prop:GeneralMapF} we know that the transposition map may be put in its operator-sum representation
\begin{equation}
	\theta (a) = \sum_{i,j=1}^{d^2} \theta_{ij} F_i a F_j
\end{equation}
for matrix $[\theta_{ij}]\in\matr{d^2}$ calculated from formula \eqref{eq:tCoeffFormula}, where $\hat{T}$ is chosen as a matricial representation of $\theta$ under the vectorization scheme. It is not difficult to show that general structure of $\hat{T}$ is
\begin{equation}
	\hat{T} = \sum_{i,j=1}^{d} E_{ij} \otimes E_{ji}
\end{equation}
where $E_{ij}$ are matrix units. $\hat{T}$ then consists of $d^2$ square blocks containing only single 1 at some location and 0s elsewhere and in fact is a permutation matrix (in literature, those are sometimes called \emph{SWAP matrices}). As an example, below we demonstrate appropriate matrices for $d=2$ and $3$:
\begin{subequations}
\begin{equation}
	\hat{T}_{2\times 2} = \left(\begin{array}{cc|cc} 1 & 0 & 0 & 0 \\ 0 & 0 & 1 & 0 \\\hline 0 & 1 & 0 & 0 \\ 0 & 0 & 0 & 1\end{array}\right),
\end{equation}
\begin{equation}
	\hat{T}_{3\times 3} = \left(
\begin{array}{ccc|ccc|ccc}
 1 & 0 & 0 & 0 & 0 & 0 & 0 & 0 & 0 \\
 0 & 0 & 0 & 1 & 0 & 0 & 0 & 0 & 0 \\
 0 & 0 & 0 & 0 & 0 & 0 & 1 & 0 & 0 \\\hline
 0 & 1 & 0 & 0 & 0 & 0 & 0 & 0 & 0 \\
 0 & 0 & 0 & 0 & 1 & 0 & 0 & 0 & 0 \\
 0 & 0 & 0 & 0 & 0 & 0 & 0 & 1 & 0 \\\hline
 0 & 0 & 1 & 0 & 0 & 0 & 0 & 0 & 0 \\
 0 & 0 & 0 & 0 & 0 & 1 & 0 & 0 & 0 \\
 0 & 0 & 0 & 0 & 0 & 0 & 0 & 0 & 1 \\
\end{array}
\right).
\end{equation}
\end{subequations}
Now, by Hermiticity of $F_i$ we have
\begin{align}
	\theta_{ij} &= \tr{\left[ (F_i \otimes \overline{F_{j}})^\hadj \hat{T} \right]} = \tr{\left[ (F_i \otimes F_{j}^\transpose) \hat{T} \right]} \\
	&= \sum_{k,l=1}^{d} \tr{\left( F_i E_{kl} \otimes F_{j}^\transpose E_{lk}\right)} = \sum_{k,l=1}^{d} \tr{F_i E_{kl}} \cdot \tr{F_{j}^\transpose E_{lk}} \nonumber \\
	&= \sum_{k,l=1}^{d} \hsiprod{F_i}{E_{kl}} \hsiprod{E_{kl}}{F_{j}^{\transpose}} = \hsiprod{F_i}{\sum_{k,l=1}^{d}\hsiprod{E_{kl}}{F_{j}^{\transpose}}E_{kl}}\nonumber \\
	&= \hsiprod{F_i}{F_{j}^{\transpose}},\nonumber
\end{align}
since canonical basis $\{E_{ij}\}$ is yet another (nonhermitian) Hilbert-Schmidt orthonormal basis. Notice that $F_{j}^{\transpose} = \pm F_j$ depending on symmetry of $F_j$ and so
\begin{equation}
	\theta_{ij} = \pm \delta_{ij}
\end{equation}
and matrix $[\theta_{ij}]$ is diagonal, $\theta_{ij} = \operatorname{diag}{\{\theta_i\}}$ for $\theta_i = \pm 1$. If $1 \leqslant i \leqslant \frac{1}{2}d(d-1)$, i.e.~$F_i$ is symmetric, we have $\theta_i = 1$; if, on the other hand $\frac{1}{2}d(d-1) + 1 \leqslant i \leqslant d^2 -d$, i.e.~$F_i$ is antisymmetric, we have $\theta_i = -1$. In the remaining case $d^2 - d + 1 \leqslant i \leqslant d^2$ the resulting diagonal matrices $F_i$ are naturally also symmetric, so we still have $\theta_i = 1$, as claimed.
\end{proof}
\begin{proposition}
The following statements hold:
\begin{enumerate}
	\item \label{TranspositionPropOne} For every (not necessarily Hermitian) Hilbert-Schmidt basis $\{G_i\}$ there exists such a Hermitian matrix $[m_{ij}]\in\matr{d^2}$ unitarily equivalent to $\hat{\theta}$ \eqref{eq:ThetaMatrix} that the mapping $a \mapsto \sum_{i,j=1}^{d^2} m_{ij} G_i a G_{j}^{\hadj}$ is a transposition.
	\item \label{TranspositionPropTwo} For every matrix $[m_{ij}] \in \matr{d^2}$ unitarily equivalent to matrix $\hat{\theta}$ \eqref{eq:ThetaMatrix} there exists such a (not necessarily Hermitian) Hilbert-Schmidt basis $\{G_i\}$ that a mapping $a \mapsto \sum_{i,j=1}^{d^2} m_{ij} G_i a G_{j}^{\hadj}$ is a transposition.
\end{enumerate}
\end{proposition}
\begin{proof}
Ad (\ref{TranspositionPropOne}). Let $\{G_i\}$ be some orthonormal Hilbert-Schmidt basis. Then there exists a unitary transformation matrix $U = [u_{ij}] \in \matr{d^2}$ such that
\begin{equation}
	G_i = \sum_{j} u_{ji} F_j \quad \text{and} \quad F_i = \sum_{j} \overline{u_{ij}} G_j .
\end{equation}
Set a matrix $[m_{ij}]$ as
\begin{equation}
	[m_{ij}] = U^\hadj \hat{\theta} U, \quad m_{ij} = \sum_{kl} \theta_k \delta_{kl} \overline{u_{ki}} u_{lj},
\end{equation}
which then yields, for $a \in \matrd$,
\begin{equation}
	\sum_{ij} m_{ij} G_i a G_{j}^{\hadj} = \sum_{i} \theta_i F_i a F_i = a^\transpose
\end{equation}
after easy algebra. Ad (\ref{TranspositionPropTwo}). Analogously, let again $[m_{ij}] = U^\hadj \hat{\theta} U$ for some arbitrarily chosen unitary $U = [u_{ij}]$. Then, if one \emph{defines} $G_i = \sum_{j} u_{ji} F_j$ then immediately we have $\sum_{ij} m_{ij} G_i a G_{j}^{\hadj} = \sum_{i} \theta_i F_i a F_i = a^\transpose$ and there exists such a basis.
\end{proof}

\subsubsection{Some properties of decomposable maps}
\label{app:DecomposableMaps}

\begin{proposition}
\label{prop:ExpOfCPmap}
Let $\phi \in \cpe{\matrd}$. Then $e^\phi \in \cpe{\matrd}$ as well.
\end{proposition}

\begin{proof}
Recall that, since $\phi$ may be represented as a complex square matrix of size $d^2$, one can always express $e^{\phi}$ as a limit
\begin{equation}\label{eq:eToPhi}
	e^{\phi} = \lim_{n\to\infty} \left( \id{} + \frac{1}{n}\phi\right)^{n},
\end{equation}
where all maps of a form $\left(\id{} + \frac{1}{n}\phi\right)^{n}$, $n\in\naturals$, are also CP. Then, the limit also defines a CP map since the cone $\cpe{\matrd}$ is closed.
\end{proof}

\begin{proposition}
\label{prop:ExpOfcoCPmap}
Let $\varphi\in\cocpe{\matrd}$. Then $e^{\varphi}\in\dece{\matrd}$.
\end{proposition}

\begin{proof}
Let $\varphi = \theta \circ \phi$ for $\phi\neq 0$ completely positive (case $\phi = 0$ gives $e^{\theta\circ\phi} = \id{}$ which is trivially decomposable). Then, it suffices to express $e^{\varphi}$ by putting $\theta\circ\phi$ in place of $\phi$ in formula \eqref{eq:eToPhi} and to notice that all maps under the limit are decomposable, for all $n\in\naturals$, as is the limit itself by the fact, that $\dece{\matrd}$ is closed.
\end{proof}

\subsection{Secondary lemmas and proofs}
\label{app:AdditionalResults}

\begin{lemma}\label{lemma:OmegaAlternateForm}
Geometric tensor $\hat{\mathbf{\Omega}}$ may be re-expressed in a form
\begin{equation}\label{eq:OmegaWithTranspositions}
	\Omega_{\mu\nu}^{jk} = \hsiprod{F_{k}^{\transpose} F_\mu}{F_{\nu}^{\transpose}F_j}.
\end{equation}
\end{lemma}

\begin{proof}
Recall that the operator-sum representation \eqref{eq:ThetaRep} of transposition map may be rearranged in a form of formula \eqref{eq:ThetaAlternateForm},
\begin{equation}
	\theta(a) = \sum_{i=1}^{d^2} F_i a F_i - 2 \sum_{i\in\mathcal{J}} F_i a F_i,
\end{equation}
where $\mathcal{J} = \{1+d(d-1)/2, \, ... \, , \, d(d-1)\}$ enumerates the \emph{antisymmetric} part of basis, i.e.~a linear span of $\{F_i : i\in\mathcal{J}\}$ is the subspace $\matrd_{\mathrm{as.}}$ of all antisymmetric matrices in $\matrd$. This fact implies that $\Omega_{\mu\nu}^{jk}$ may be, after using cyclicity of trace, put in a form
\begin{align}\label{eq:OmegaProjectionForm}
	\Omega_{\mu\nu}^{jk} &= \hsiprod{F_\mu F_j}{\sum_{i=1}^{d^2} \theta_i \hsiprod{F_i}{F_\nu F_k}F_i} \\
	&= \hsiprod{F_\mu F_j}{ \left( \id{\matrd} - 2 \proj{\mathrm{as.}} \right)(F_\nu F_k)} \nonumber
\end{align}
where $\id{\matrd}$ is the identity map on $\matrd$ and $\proj{\mathrm{as.}}$ is the orthogonal projection onto $\matrd_{\mathrm{as.}}$ given as
\begin{equation}
	\proj{\mathrm{as.}}(a) = \sum_{i\in\mathcal{J}} \hsiprod{F_i}{a}F_i .
\end{equation}
Let $\{e_i\}$ be a canonical basis in $\complexes^d$. By dimension count, it is easy to see that space $\matrd$ may be identified with a Hilbert space tensor product $\complexes^d \otimes \complexes^d$, with a mapping $\zeta : \complexes^d \otimes \complexes^d \to \matrd$ defined by its action on basis elements as
\begin{equation}
	\zeta(e_i \otimes e_j) = E_{ij} = \left| e_i \right\rangle \left\langle e_j \right|
\end{equation}
and then extended by linearity, being a natural bijection. Under action of $\zeta$, every vector $x = \sum_{ij} x_{ij} e_i \otimes e_j \in \complexes^d \otimes \complexes^d$ can be isomorphically represented as a matrix $[x_{ij}]\in\matrd$ and vice versa. This implies, that $\matrd_{\mathrm{as.}}$ is identified with $ \complexes^d \wedge \complexes^d$, the antisymmetric subspace of $\complexes^d \otimes \complexes^d$. In result, operator $P_{\mathrm{as.}} = \zeta^{-1} \circ \proj{\mathrm{as.}} \circ \zeta$ is the corresponding projection onto $\complexes^d \wedge \complexes^d$. We know however, that such projection may be expressed in a form
\begin{equation}
	P_{\mathrm{as.}} = \frac{1}{2} \left( \id{\complexes^d \otimes \complexes^d} - V \right),
\end{equation}
with $V$ being the \emph{swap operator} on $\complexes^d \otimes \complexes^d$ defined via
\begin{equation}
	V(x \otimes y) = y\otimes x, \quad x,y \in \complexes^d.
\end{equation}
From this, we have
\begin{equation}
	\id{\matrd} - 2 \proj{\mathrm{as.}} = \zeta\circ V \circ \zeta^{-1},
\end{equation}
which by direct check is a \emph{transposition} on $\matrd$. In result, \eqref{eq:OmegaProjectionForm} reads
\begin{equation}
	\Omega_{\mu\nu}^{jk} = \hsiprod{F_\mu F_j}{ (F_\nu F_k)^{\transpose}}
\end{equation}
which is equal to claimed form \eqref{eq:OmegaWithTranspositions} after easy manipulations.
\end{proof}

\begin{lemma}
\label{lemma:OmegaMatrixPositive}
Matrix $[\omega_{jk}(t)] \in \matr{d^2}$ given via expression
\begin{equation}
	\omega_{jk}(t) = \lim_{\epsilon\searrow 0}\frac{1}{\epsilon}  y_{jk}(t+\epsilon, t)
\end{equation}
is well-defined and positive semidefinite for all $t\in[t_1,t_2]$.
\end{lemma}

\begin{proof}
Let $\vec{v} \in \complexes^{d^2}$ and define function $f_{\vec{v}} : [t_1,t_2]^2 \to \reals$ as
\begin{equation}
	f_{\vec{v}}(t, s) = \iprod{\vec{v}}{[y_{jk}(t,s)]\vec{v}} = \sum_{j,k=1}^{d^2} y_{jk}(t,s) v_j \overline{v_k}.
\end{equation}
Since matrix $[y_{jk}(t,s)]$ was uniquely identified with a CP map $Y_{t,s}$ appearing in the propagator, it is positive semidefinite for all $t\geqslant s$, so clearly $f_{\vec{v}}(t,s) \geqslant 0$ for every $\vec{v}\in\complexes^{d^2}$ and $t\geqslant s$. Moreover, from proposition \ref{prop:Vproperties} we have $Y_{t,t} = 0$ and so $f_{\vec{v}}(t,t) = 0$. Let then $t \in [t_1, t_2]$ be arbitrary and assume indirectly, that $f_{\vec{v}}(\cdot \, , t)$ is decreasing in some interval $[t_0, \xi_0]$ for some $t_0 \geqslant t$. Then there exists $\xi \geqslant t_0$ such that $f(\xi,t_0) < f(t_0, t_0) = 0$, which is a contradiction. This yields that $\xi \mapsto f_{\vec{v}}(\xi , t)$, where $\xi \geqslant t$, must be non-decreasing for every $t\geqslant 0$. We will use this reasoning in a following computation. The formula for matrix $[\omega_{jk}(t)]$ can be rewritten as
\begin{align}
	\omega_{jk}(t) &= \lim_{\epsilon\searrow 0}\frac{1}{\epsilon}  y_{jk}(t+\epsilon, t) \\
	&= \lim_{\epsilon\searrow 0}\frac{y_{jk}(t+\epsilon, t) - y_{jk}(t,t)}{\epsilon} = \left.\frac{\partial y_{jk}(\xi, t)}{\partial \xi}\right|_{t}, \nonumber
\end{align}
since $y_{jk}(t,t) = 0$, i.e.~as a derivative wrt.~first variable of a matrix $[y_{jk}(\xi,t)]$, computed at $\xi = t$. This however yields, for every $\vec{v}\in\complexes^{d^2}$,
\begin{equation}
	\sum_{j,k=1}^{d^2} \omega_{\mu\nu}(t) v_j \overline{v_k} = \sum_{j,k=1}^{d^2} \left.\frac{\partial y_{jk}(\xi, t)}{\partial \xi}\right|_{t} v_j \overline{v_k} = \left.\frac{\partial f_{\vec{v}}(\xi,t)}{\partial \xi}\right|_{t} \geqslant 0
\end{equation}
due to demonstrated monotonicity of $f_{\vec{v}}(\cdot \, , t )$. This shows that $[\omega_{jk}(t)] \in \matrd^{+}$ for $t\in [t_1, t_2]$.
\end{proof}

\subsection{Derivation of formula (\ref{eq:NtFormula})}
\label{app:NtDerivation}
Starting with expression \eqref{eq:Nt} for $N_t$ we rewrite it by expanding the anticommutator and expressing $\eta_{\mu\nu}(t)$ as \eqref{eq:EtaMatrixOmega},
\begin{align}\label{eq:NtExpanded}
	N_t (\rho) = &\sum_{j,k=1}^{d^2} \sum_{\mu,\nu=1}^{d^2-1} \Omega_{\mu\nu}^{jk} \omega_{jk}(t) F_\mu \rho F_\nu - \frac{1}{2} \sum_{j,k=1}^{d^2}\sum_{\mu,\nu=1}^{d^2-1} \Omega_{\mu\nu}^{jk} \omega_{jk}(t) F_\nu F_\mu \rho \\
	&- \frac{1}{2} \sum_{j,k=1}^{d^2}\sum_{\mu,\nu=1}^{d^2-1} \Omega_{\mu\nu}^{jk} \omega_{jk}(t) \rho F_\nu F_\mu - i\comm{K_t}{\rho}\nonumber ,
\end{align}
where we also put back $\Omega_{\mu\nu}^{jk}$ as given in \eqref{eq:OmegaTensor} in each term. In order to reintroduce the transposition map $\theta$ into the expression, we expand the summations over $\mu$, $\nu$ up to $d^2$ and then subtract redundant terms. The first term appearing at the right hand side of equality \eqref{eq:NtExpanded} is therefore
\begin{align}\label{eq:NtExpanded2}
	&\sum_{j,k=1}^{d^2} \sum_{\mu,\nu=1}^{d^2-1} \Omega_{\mu\nu}^{jk} \omega_{jk}(t) F_\mu \rho F_\nu \\
	&= \sum_{jklm} \theta_{l} \omega_{jk}(t) \left( \sum_{\mu}\xi_{lj\mu}F_\mu - \xi_{ljd^2}F_{d^2}\right) \rho \left(\sum_{\nu} \xi_{lk\nu}F_\nu - \xi_{lkd^2}F_{d^2}\right)^{\hadj} \nonumber \\
	&= \sum_{jklm} \theta_{l} \omega_{jk}(t) \left( F_l F_j - \frac{1}{\sqrt{d}}\xi_{ljd^2}\right) \rho \left(F_l F_k - \frac{1}{\sqrt{d}}\xi_{lkd^2}\right)^{\hadj} \nonumber \\
	&= \sum_{jk} \omega_{jk}(t) \left[ (F_j \rho F_k)^\transpose - A_k F_j \rho - \rho F_k A_{j}^{\hadj} + b_{jk} \rho \right] \nonumber
\end{align}
for quantities
\begin{equation}
	A_k = \sum_{l} \theta_{l} \overline{\xi_{lkd^2}} F_l, \quad b_{jk} = \sum_{l} \theta_{l} \xi_{jld^2} \overline{\xi_{lkd^2}},
\end{equation}
where we employed composition rule \eqref{eq:FiFjCompRule} and operator sum representation \eqref{eq:ThetaRep} of transposition map $\theta$ (all ``limitless'' summation indices run from 1 up to $d^2$). Next, we utilize the fact that $[\omega_{jk}(t)]$ was a positive semi-definite matrix for all $t$, i.e.~we introduce
\begin{equation}
	\omega_{jk}(t) = \sum_\alpha c_{j\alpha}(t) \overline{c_{k\alpha}(t)},
\end{equation}
for some matrix $[c_{jk}(t)]$. This, inserted into the last line of \eqref{eq:NtExpanded2} allows to re-express it as
\begin{align}\label{eq:FirstTerm}
	\sum_{j,k=1}^{d^2} \sum_{\mu,\nu=1}^{d^2-1} &\Omega_{\mu\nu}^{jk} \omega_{jk}(t) F_\mu \rho F_\nu \\
	&= \sum_\alpha \left[ (C_{\alpha,t} \rho C_{\alpha,t}^{\hadj})^\transpose - D_t \rho - \rho D_{t}^{\hadj} + e(t) \rho \right], \nonumber
\end{align}
where we defined
\begin{equation}
	C_{\alpha,t} = \sum_j c_{j\alpha}(t) F_j, \quad D_t = \sum_{jk} \omega_{jk}(t) A_k F_j , \quad e(t) = \sum_{jk} \omega_{jk}(t) b_{jk}
\end{equation}
and applied Hermiticity of $[\omega_{jk}(t)]$ in order to get the $\rho D_{t}^{\hadj}$ term. Now, we notice that the two remaining terms at the right hand side of equality \eqref{eq:NtExpanded} have essentially the same structure and differ from the first term only by order of matrices $\rho$, $F_\mu$ and $F_\nu$ so they can be transformed by applying nearly exactly the same steps. After some easy algebra, we obtain
\begin{equation}\label{eq:SecondTerm}
	\sum_{j,k=1}^{d^2}\sum_{\mu,\nu=1}^{d^2-1} \Omega_{\mu\nu}^{jk} \omega_{jk}(t) F_\nu F_\mu \rho = \sum_\alpha \left[ C_{\alpha,t}^{\hadj} C_{\alpha,t} - D_{t}^{\hadj} - D_{t} + e(t) \right]\rho
\end{equation}
for the second term, as well as
\begin{equation}\label{eq:ThirdTerm}
	\sum_{j,k=1}^{d^2}\sum_{\mu,\nu=1}^{d^2-1} \Omega_{\mu\nu}^{jk} \omega_{jk}(t) \rho F_\nu F_\mu = \rho \sum_\alpha \left[ C_{\alpha,t}^{\hadj} C_{\alpha,t} - D_{t}^{\hadj} - D_{t} + e(t) \right]
\end{equation}
for the third one. Now, we insert \eqref{eq:FirstTerm}, \eqref{eq:SecondTerm} and \eqref{eq:ThirdTerm} back into \eqref{eq:NtExpanded} which becomes
\begin{align}
	N_t (\rho) = &-i \comm{K_t}{\rho} + \left( \sum_\alpha C_{\alpha,t} \rho C_{\alpha,t}^{\hadj}\right)^\transpose - \frac{1}{2}\comm{D_t - D_{t}^{\hadj}}{\rho} \\
	&- \frac{1}{2} \sum_\alpha \acomm{C_{\alpha,t}^{\hadj}C_{\alpha,t}}{\rho} \nonumber
\end{align}
after some effort.


\end{document}